\documentclass[11pt]{article}
\usepackage{epsf}
\usepackage{amsmath}
\usepackage{epsfig}
\usepackage{times}
\usepackage{amssymb}
\usepackage{amsthm}
\usepackage{setspace}
\usepackage{cite}

\usepackage{algorithmic}  
\usepackage{algorithm}

\usepackage{shadow}
\usepackage{fancybox}
\usepackage{fancyhdr}

\def\w{{\bf w}}

\def\y{{\bf y}}

\def\x{{\bf x}}

\def\x{{\mathbf x}}

\def\w{{\bf w}}

\def\x{{\bf x}}
\def\y{{\bf y}}

\def\b{{\bf b}}

\def\h{{\bf h}}

\def\be{\begin{equation}}
\def\ee{\end{equation}}
\def\ba{\left[\begin{array}}
\def\ea{\end{array}\right]}

\def\t{{\bf t}}

\def\w{{\bf w}}

\def\x{{\bf x}}
\def\y{{\bf y}}

\def\b{{\bf b}}

\def\1{{\bf 1}}

\def\g{{\bf g}}
\def\0{{\bf 0}}

\def\erfinv{\mbox{erfinv}}
\def\htheta{\hat{\theta}}

\def\betasec{\beta_{sec}}

\def\Ssec{S_{sec}}

\def\hthetasec{\hat{\theta}_{sec}}
\def\thetasec{\theta_{sec}}

\def\Sstr{S_{str}}
\def\betastr{\beta_{str}}

\def\Sstrnon{S_{str}^{+}}

\def\betawnon{\beta_{w}^{+}}

\def\hthetawnon{\hat{\theta}_{w}^+}
\def\thetawnon{\theta_{w}^+}

\newtheorem{theorem}{Theorem}
\newtheorem{corollary}{Corollary}

\newtheorem{lemma}{Lemma}

\begin{document}

\begin{singlespace}

\title {Lifting $\ell_1$-optimization strong and sectional thresholds 
}
\author{
\textsc{Mihailo Stojnic}
\\
\\
{School of Industrial Engineering}\\
{Purdue University, West Lafayette, IN 47907} \\
{e-mail: {\tt mstojnic@purdue.edu}} }
\date{}
\maketitle

\centerline{{\bf Abstract}} \vspace*{0.1in}

In this paper we revisit under-determined linear systems of equations with sparse solutions. As is well known, these systems are among core mathematical problems of a very popular compressed sensing field. The popularity of the field as well as a substantial academic interest in linear systems with sparse solutions are in a significant part due to seminal results \cite{CRT,DonohoPol}. Namely, working in a statistical scenario, \cite{CRT,DonohoPol} provided substantial mathematical progress in characterizing relation between the dimensions of the systems and the sparsity of unknown vectors recoverable through a particular polynomial technique called $\ell_1$-minimization. In our own series of work \cite{StojnicCSetam09,StojnicUpper10,StojnicEquiv10} we also provided a collection of mathematical results related to these problems. While, Donoho's work \cite{DonohoPol,DonohoUnsigned} established (and our own work \cite{StojnicCSetam09,StojnicUpper10,StojnicEquiv10} reaffirmed) the typical or the so-called \emph{weak threshold} behavior of $\ell_1$-minimization many important questions remain unanswered. Among the most important ones are those that relate to non-typical or the so-called \emph{strong threshold} behavior. These questions are usually combinatorial in nature and known techniques come up short of providing the exact answers. In this paper we provide a powerful mechanism that that can be used to attack the ``tough" scenario, i.e. the \emph{strong threshold} (and its a similar form called \emph{sectional threshold}) of $\ell_1$-minimization. The results we present offer substantial conceptual (in some cases even practical) improvement over known counterparts from \cite{StojnicCSetam09}. Moreover, they emphasize the hardness of the underlying problems caused by their combinatorial structure. Along the same lines, the results that we will present together with those from \cite{StojnicMoreSophHopBnds10} in a way also provide a substantial breakthrough that can be utilized for studying a huge number of other hard combinatorial problems.

\vspace*{0.25in} \noindent {\bf Index Terms: $\ell_1$-minimization; strong, sectional threshold}.

\end{singlespace}

\section{Introduction}
\label{sec:back}

In this paper we will revisit the under-determined linear systems of equations with sparse solutions. While linear systems have been known for a long time as one of most basic classical mathematical problems, they gained a particular attention over last the decade in significant part due to popularity of a field called compressed sensing. Our interest in this paper is a purely mathematical study of linear systems so we will refrain from a detailed presentation of the compressed sensing ideas. However, to insure a bare minimum of completeness we just briefly sketch the basic idea. In contexts where the compressed sensing is to be employed  it typically happens that signals of interest can be represented as sparse vectors in certain basis and one then given the signal and basis ends up having to recover the representation coefficients which in essence establishes a linear system of equations. Since a few of the coefficients are expected to be nonzero one then effectively has a linear system which is known to have a sparse solution. The goal is to determine such a solution and to do so in systems dimensions hopefully as small as mathematically necessary (way more about the compressed sensing conception and various problems of interest within the fields that grew out of the above basic compressed sensing concept can be found in a tone of references; here we point out to a couple of introductory papers, e.g. \cite{DOnoho06CS,CRT}).

Going back to the mathematics of linear systems, we start by recalling precise mathematical descriptions of such problems. As mentioned above, it essentially boils down to finding sparse solutions of under-determined systems of linear equations. In a more precise mathematical language we would like to find a $k$-sparse $\x$ such
that
\begin{equation}
A\x=\y \label{eq:system}
\end{equation}
where $A$ is an $m\times n$ ($m<n$) matrix and $\y$ is
an $m\times 1$ vector (here and in the rest of the paper, under $k$-sparse vector we assume a vector that has at most $k$ nonzero
components). Of course, the assumption will be that such an $\x$ exists.

To make writing in the rest of the paper easier, we will assume the
so-called \emph{linear} regime, i.e. we will assume that $k=\beta n$
and that the number of equations is $m=\alpha n$ where
$\alpha$ and $\beta$ are constants independent of $n$ (more
on the non-linear regime, i.e. on the regime when $m$ is larger than
linearly proportional to $k$ can be found in e.g.
\cite{CoMu05,GiStTrVe06,GiStTrVe07}).

A particularly successful technique for solving (\ref{eq:system}) is a linear programming relaxation called $\ell_1$-optimization. (Variations of the standard $\ell_1$-optimization from e.g.
\cite{CWBreweighted,SChretien08,SaZh08}) as well as those from \cite{SCY08,FL08,GN03,GN04,GN07,DG08} related to $\ell_q$-optimization, $0<q<1$
are possible as well.) Basic $\ell_1$-optimization algorithm finds $\x$ in
(\ref{eq:system}) by solving the following $\ell_1$-norm minimization problem
\begin{eqnarray}
\mbox{min} & & \|\x\|_{1}\nonumber \\
\mbox{subject to} & & A\x=\y. \label{eq:l1}
\end{eqnarray}
Due to its popularity the literature on the use of the above algorithm is rapidly growing. We below restrict our attention to two, in our mind, the most influential works that relate to (\ref{eq:l1}).

The first one is \cite{CRT} where the authors were able to show that if
$\alpha$ and $n$ are given, $A$ is given and satisfies the restricted isometry property (RIP) (more on this property the interested reader can find in e.g. \cite{Crip,CRT,Bar,Ver,ALPTJ09}), then
any unknown vector $\x$ with no more than $k=\beta n$ (where $\beta$
is a constant dependent on $\alpha$ and explicitly
calculated in \cite{CRT}) non-zero elements can be recovered by
solving (\ref{eq:l1}). As expected, this assumes that $\y$ was in
fact generated by that $\x$ and given to us.

However, the RIP is only a \emph{sufficient}
condition for $\ell_1$-optimization to produce the $k$-sparse solution of
(\ref{eq:system}). Instead of characterizing $A$ through the RIP
condition, in \cite{DonohoUnsigned,DonohoPol} Donoho looked at its geometric properties/potential. Namely,
in \cite{DonohoUnsigned,DonohoPol} Donoho considered the polytope obtained by
projecting the regular $n$-dimensional cross-polytope $C_p^n$ by $A$. He then established that
the solution of (\ref{eq:l1}) will be the $k$-sparse solution of
(\ref{eq:system}) if and only if
$AC_p^n$ is centrally $k$-neighborly
(for the definitions of neighborliness, details of Donoho's approach, and related results the interested reader can consult now already classic references \cite{DonohoUnsigned,DonohoPol,DonohoSigned,DT}). In a nutshell, using the results
of \cite{PMM,AS,BorockyHenk,Ruben,VS}, it is shown in
\cite{DonohoPol}, that if $A$ is a random $m\times n$
ortho-projector matrix then with overwhelming probability $AC_p^n$ is centrally $k$-neighborly (as usual, under overwhelming probability we in this paper assume
a probability that is no more than a number exponentially decaying in $n$ away from $1$). Miraculously, \cite{DonohoPol,DonohoUnsigned} provided a precise characterization of $m$ and $k$ (in a large dimensional context) for which this happens.

It should be noted that one usually considers success of
(\ref{eq:l1}) in recovering \emph{any} given $k$-sparse $\x$ in (\ref{eq:system}). It is also of interest to consider success of
(\ref{eq:l1}) in recovering
\emph{almost any} given $\x$ in (\ref{eq:system}). We below make a distinction between these
cases and recall on some of the definitions from
\cite{DonohoPol,DT,DTciss,DTjams2010,StojnicCSetam09,StojnicICASSP09}.

Clearly, for any given constant $\alpha\leq 1$ there is a maximum
allowable value of $\beta$ such that for \emph{any} given $k$-sparse $\x$ in (\ref{eq:system}) the solution of (\ref{eq:l1})
is with overwhelming probability exactly that given $k$-sparse $\x$. We will refer to this maximum allowable value of
$\beta$ as the \emph{strong threshold} (see
\cite{DonohoPol}) and will denote it as $\beta_{str}$. Similarly, for any given constant
$\alpha\leq 1$ and \emph{any} given $\x$ with a given fixed location of non-zero components and a given fixed combination of its elements signs
there will be a maximum allowable value of $\beta$ such that
(\ref{eq:l1}) finds that given $\x$ in (\ref{eq:system}) with overwhelming
probability. We will refer to this maximum allowable value of
$\beta$ as the \emph{weak threshold} and will denote it by $\beta_{w}$ (see, e.g. \cite{StojnicICASSP09,StojnicCSetam09}). One can also go a step further and consider scenario where for any given constant
$\alpha\leq 1$ and \emph{any} given $\x$ with a given fixed location of non-zero components
there will be a maximum allowable value of $\beta$ such that
(\ref{eq:l1}) finds that given $\x$ in (\ref{eq:system}) with overwhelming
probability. We will refer to such a $\beta$ as the \emph{sectional threshold} and will denote it by $\beta_{sec}$ (more on the definition of the sectional threshold the interested reader can find in e.g. \cite{DonohoPol,StojnicCSetam09}).

When viewed within this frame the results of \cite{CRT,DOnoho06CS} established that $\ell_1$-minimization achieves recovery through a linear scaling of all important dimensions ($k$, $m$, and $n$). Moreover, for all $\beta$'s defined above lower bounds were provided in \cite{CRT}. On the other hand, the results of \cite{DonohoPol,DonohoUnsigned} established the exact values of $\beta_w$ and provided lower bounds on $\beta_{str}$ and $\beta_{sec}$.

In a series of our own work (see, e.g. \cite{StojnicICASSP09,StojnicCSetam09,StojnicUpper10}) we then created an alternative probabilistic approach which was capable of providing the precise characterization of $\beta_w$ as well and thereby reestablishing the results of Donoho \cite{DonohoPol} through a purely probabilistic approach. We also presented in \cite{StojnicCSetam09} further results related to lower bounds on $\beta_{str}$ and $\beta_{sec}$. Below, we will present the three theorems that we proved in  \cite{StojnicICASSP09,StojnicCSetam09,StojnicUpper10} and are related to $\beta_w$, $\beta_{sec}$, and $\beta_{str}$. We find it useful for the ease of the presentation that will follow to have these theorems clearly restated at one place. Fairly often we will use these theorems as a benchmark for the results that we will present in this paper.

The first of the theorems relates to the weak threshold $\beta_w$.

\begin{theorem}(Weak threshold -- exact \cite{StojnicCSetam09,StojnicUpper10})
Let $A$ be an $m\times n$ matrix in (\ref{eq:system})
with i.i.d. standard normal components. Let
the unknown $\x$ in (\ref{eq:system}) be $k$-sparse. Further, let the location and signs of nonzero elements of $\x$ be arbitrarily chosen but fixed.
Let $k,m,n$ be large
and let $\alpha=\frac{m}{n}$ and $\beta_w=\frac{k}{n}$ be constants
independent of $m$ and $n$. Let $\erfinv$ be the inverse of the standard error function associated with zero-mean unit variance Gaussian random variable.  Further,
let all $\epsilon$'s below be arbitrarily small constants.
\begin{enumerate}
\item Let $\htheta_w$, ($\beta_w\leq \htheta_w\leq 1$) be the solution of
\begin{equation}
(1-\epsilon_{1}^{(c)})(1-\beta_w)\frac{\sqrt{\frac{2}{\pi}}e^{-(\erfinv(\frac{1-\theta_w}{1-\beta_w}))^2}}{\theta_w}-\sqrt{2}\erfinv ((1+\epsilon_{1}^{(c)})\frac{1-\theta_w}{1-\beta_w})=0.\label{eq:thmweaktheta}
\end{equation}
If $\alpha$ and $\beta_w$ further satisfy
\begin{equation}
\alpha>\frac{1-\beta_w}{\sqrt{2\pi}}\left (\sqrt{2\pi}+2\frac{\sqrt{2(\erfinv(\frac{1-\htheta_w}{1-\beta_w}))^2}}{e^{(\erfinv(\frac{1-\htheta_w}{1-\beta_w}))^2}}-\sqrt{2\pi}
\frac{1-\htheta_w}{1-\beta_w}\right )+\beta_w
-\frac{\left ((1-\beta_w)\sqrt{\frac{2}{\pi}}e^{-(\erfinv(\frac{1-\hat{\theta}_w}{1-\beta_w}))^2}\right )^2}{\hat{\theta}_w}\label{eq:thmweakalpha}
\end{equation}
then with overwhelming probability the solution of (\ref{eq:l1}) is the $k$-sparse $\x$ from (\ref{eq:system}).
\item Let $\htheta_w$, ($\beta_w\leq \htheta_w\leq 1$) be the solution of
\begin{equation}
(1+\epsilon_{2}^{(c)})(1-\beta_w)\frac{\sqrt{\frac{2}{\pi}}e^{-(\erfinv(\frac{1-\theta_w}{1-\beta_w}))^2}}{\theta_w}-\sqrt{2}\erfinv ((1-\epsilon_{2}^{(c)})\frac{1-\theta_w}{1-\beta_w})=0.\label{eq:thmweaktheta1}
\end{equation}
If on the other hand $\alpha$ and $\beta_w$ satisfy
\begin{multline}
\hspace{-.5in}\alpha<\frac{1}{(1+\epsilon_{1}^{(m)})^2}\left ((1-\epsilon_{1}^{(g)})(\htheta_w+\frac{2(1-\beta_w)}{\sqrt{2\pi}} \frac{\sqrt{2(\erfinv(\frac{1-\htheta_w}{1-\beta_w}))^2}}{e^{(\erfinv(\frac{1-\htheta_w}{1-\beta_w}))^2}})
-\frac{\left ((1-\beta_w)\sqrt{\frac{2}{\pi}}e^{-(\erfinv(\frac{1-\hat{\theta}_w}{1-\beta_w}))^2}\right )^2}{\hat{\theta}_w(1+\epsilon_{3}^{(g)})^{-2}}\right )\label{eq:thmweakalpha}
\end{multline}
then with overwhelming probability there will be a $k$-sparse $\x$ (from a set of $\x$'s with fixed locations and signs of nonzero components) that satisfies (\ref{eq:system}) and is \textbf{not} the solution of (\ref{eq:l1}).
\end{enumerate}
\label{thm:thmweakthr}
\end{theorem}
\begin{proof}
The first part was established in \cite{StojnicCSetam09} and the second one was established in \cite{StojnicUpper10}. An alternative way of establishing the same set of results was also presented in \cite{StojnicEquiv10}. Of course, the weak thresholds were first computed in \cite{DonohoPol} through a different geometric approach.
\end{proof}

We below provide a more informal interpretation of what was established by the above theorem. Assume the setup of the above theorem. Let $\alpha_w$ and $\beta_w$ satisfy the following:

\noindent \underline{\underline{\textbf{Fundamental characterization of the $\ell_1$ minimization weak threshold:}}}

\begin{center}
\shadowbox{$
(1-\beta_w)\frac{\sqrt{\frac{2}{\pi}}e^{-(\erfinv(\frac{1-\alpha_w}{1-\beta_w}))^2}}{\alpha_w}-\sqrt{2}\erfinv (\frac{1-\alpha_w}{1-\beta_w})=0.
$}
-\vspace{-.5in}\begin{equation}
\label{eq:thmweaktheta2}
\end{equation}
\end{center}

Then:
\begin{enumerate}
\item If $\alpha>\alpha_w$ then with overwhelming probability the solution of (\ref{eq:l1}) is the $k$-sparse $\x$ from (\ref{eq:system}).
\item If $\alpha<\alpha_w$ then with overwhelming probability there will be a $k$-sparse $\x$ (from a set of $\x$'s with fixed locations and signs of nonzero components) that satisfies (\ref{eq:system}) and is \textbf{not} the solution of (\ref{eq:l1}).
    \end{enumerate}

The following theorem summarizes the results related to the sectional threshold ($\beta_{sec}$) that we obtained in \cite{StojnicCSetam09}.

\begin{theorem}(Sectional threshold - lower bound)
Let $A$ be an $m\times n$ measurement matrix in (\ref{eq:system})
with the null-space uniformly distributed in the Grassmanian. Let
the unknown $\x$ in (\ref{eq:system}) be $k$-sparse. Further, let the location of nonzero elements of $\x$ be arbitrarily chosen but fixed.
Let $k,m,n$ be large
and let $\alpha=\frac{m}{n}$ and $\betasec=\frac{k}{n}$ be constants
independent of $m$ and $n$. Let $\erfinv$ be the inverse of the standard error function associated with zero-mean unit variance Gaussian random variable.  Further,
let $\epsilon>0$ be an arbitrarily small constant and $\hthetasec$, ($\betasec\leq \hthetasec\leq 1$) be the solution of
\begin{equation}
(1-\epsilon)(1-\betasec)\frac{\sqrt{\frac{2}{\pi}}e^{-(\erfinv(\frac{1-\thetasec}{1-\betasec}))^2}-\sqrt{\frac{2}{\pi}}\frac{\betasec}{1-\betasec}}
{\thetasec}-\sqrt{2}\erfinv ((1+\epsilon)\frac{1-\thetasec}{1-\betasec})=0.\label{eq:thmsectheta}
\end{equation}
If $\alpha$ and $\betasec$ further satisfy
\begin{equation}
\hspace{-.6in}\alpha>\frac{1-\betasec}{\sqrt{2\pi}}\left (\sqrt{2\pi}+2\frac{\sqrt{2(\erfinv(\frac{1-\hthetasec}{1-\betasec}))^2}}{e^{(\erfinv(\frac{1-\hthetasec}{1-\betasec}))^2}}-\sqrt{2\pi}
\frac{1-\hthetasec}{1-\betasec}\right )+\betasec
-\frac{\left ((1-\betasec)\sqrt{\frac{2}{\pi}}e^{-(\erfinv(\frac{1-\hthetasec}{1-\betasec}))^2}-\sqrt{\frac{2}{\pi}}\betasec\right )^2}{\hthetasec}\label{eq:thmsecalpha}
\end{equation}
then with overwhelming probability the solution of (\ref{eq:l1}) is the $k$-sparse $\x$ from (\ref{eq:system}).
\label{thm:thmsecthr}
\end{theorem}

Finally the following theorem summarizes the results related to the strong thresholds ($\beta_{s}$) that we obtained in \cite{StojnicCSetam09}.

\begin{theorem}(Strong threshold - lower bound)
Let $A$ be an $m\times n$ measurement matrix in (\ref{eq:system})
with the null-space uniformly distributed in the Grassmanian. Let
the unknown $\x$ in (\ref{eq:system}) be $k$-sparse. Let $k,m,n$ be large
and let $\alpha=\frac{m}{n}$ and $\beta_{str}=\frac{k}{n}$ be constants
independent of $m$ and $n$. Let $\erfinv$ be the inverse of the standard error function associated with zero-mean unit variance Gaussian random variable. Further,
let $\epsilon>0$ be an arbitrarily small constant and $\htheta_s$, ($\beta_{str}\leq \htheta_s\leq 1$) be the solution of
\begin{equation}
(1-\epsilon)\frac{\sqrt{\frac{2}{\pi}}e^{-(\erfinv(1-\theta_s))^2}-2\sqrt{\frac{2}{\pi}}e^{-(\erfinv(1-\beta_{str}))^2}}{\theta_s}-\sqrt{2}\erfinv ((1+\epsilon)(1-\theta_s))=0.\label{eq:thmstrtheta}
\end{equation}
If $\alpha$ and $\beta_{str}$ further satisfy
\begin{equation}
\alpha>\frac{1}{\sqrt{2\pi}}\left (\sqrt{2\pi}+2\frac{\sqrt{2(\erfinv(1-\htheta_s))^2}}{e^{(\erfinv(1-\htheta))^2}}-\sqrt{2\pi}(1-\htheta_s)\right )
-\frac{\left (\sqrt{\frac{2}{\pi}}e^{-(\erfinv(1-\htheta_s))^2}-2\sqrt{\frac{2}{\pi}}e^{-(\erfinv(1-\beta_{str}))^2}\right )^2}{\htheta_s}\label{eq:thmstralpha}
\end{equation}
then with overwhelming probability the solution of (\ref{eq:l1}) is the $k$-sparse $\x$ from (\ref{eq:system}).\label{thm:thmstrthr}
\end{theorem}

We will show the results for the sectional and strong thresholds one can obtain through the above theorems in subsequent sections when we discuss the corresponding ones obtained in this paper. As for the contribution of this paper, essentially we will develop a mechanism that can provide a substantial conceptual improvement of the sectional and strong threshold results that we obtained in \cite{StojnicCSetam09} (of course it will match the results we already obtained for the weak threshold in \cite{StojnicCSetam09}).

We organize the rest of the paper in the following way. In Section
\ref{sec:secthr} we present the core of the mechanism and how it can be used to improve the sectional thresholds results. In Section \ref{sec:strthr} we will then present a neat modification of the mechanism so that it can handle the strong thresholds as well. In Section \ref{sec:backnon} we introduce a special class of unknown vectors $\x$, namely, vectors $\x$ a priori known to have nonnegative components and present several known results related such vectors. Finally in Section \ref{sec:strthrnon} we will generalize the strong threshold results obtained for general vectors $\x$ in Section \ref{sec:strthr} to those that relate to vectors $\x$ a priori known to have only non-negative components. In Section \ref{sec:conc} we discuss obtained results and provide several conclusions related to their importance.

\section{Lifting $\ell_1$-minimization sectional threshold}
\label{sec:secthr}

In this section we look at the sectional thresholds. We do mention before presenting anything that these sectional threshold results were substantially easier to establish than the ones that will follow for the strong thresholds (of course, we by no means consider them easy to establish).

Throughout the presentation in this and all subsequent sections we will assume a substantial level of familiarity with many of the well-known results that relate to the performance characterization of (\ref{eq:l1}) (we will fairly often recall on many results/definitions that we established in \cite{StojnicCSetam09}). We start by defining a set $\Ssec$
\begin{equation}
\Ssec=\{\w\in S^{n-1}| \quad \sum_{i=n-k+1}^n |\w_i|<\sum_{i=1}^{n-k}|\w_{i}|\},\label{eq:defSsec}
\end{equation}
where $S^{n-1}$ is the unit sphere in $R^n$. Then it was established in \cite{StojnicCSetam09} that the following optimization problem is of critical importance in determining the sectional threshold of $\ell_1$-minimization
\begin{equation}
\xi_{sec}=\min_{\w\in\Ssec}\|A\w\|_2.\label{eq:negham1}
\end{equation}
Namely, what was established in \cite{StojnicCSetam09} is roughly the following: if $\xi_{sec}$ is positive with overwhelming probability for certain combination of $k$, $m$, and $n$ then for $\alpha=\frac{m}{n}$ one has a lower bound $\beta_{sec}=\frac{k}{n}$ on the true value of the sectional threshold with overwhelming probability. Also, the mechanisms of \cite{StojnicCSetam09} were powerful enough to establish the concentration of $\xi_{sec}$. This essentially means that if we can show that $E\xi_{sec}>0$ for certain $k$, $m$, and $n$ we can then obtain the lower bound on the sectional threshold. In fact, this is precisely what was done in \cite{StojnicCSetam09}. However, the results we obtained for the sectional threshold through such a consideration were not exact. The main reason of course was inability to determine $E\xi_{sec}$ exactly. Instead we resorted to its lower bounds and those turned out to be loose. In this paper we will use some of the ideas we recently introduced in \cite{StojnicMoreSophHopBnds10} to provide a substantial conceptual improvement in these bounds which would in turn reflect in a conceptual improvement of the sectional thresholds (and later on an even substantial practical improvement of all strong thresholds).

Below we present a way to create a lower-bound on the optimal value of (\ref{eq:negham1}).

\subsection{Lower-bounding $\xi_{sec}$}
\label{sec:lbxisec}

In this section we will look at problem from (\ref{eq:negham1}). As mentioned in theorems, we will assume that the elements of $A$ are i.i.d. standard normal random variables. We will also first recall on a couple of results that we obtained in \cite{StojnicMoreSophHopBnds10}. We start with the following result from \cite{Gordon85} that relates to statistical properties of certain Gaussian processes.
\begin{theorem}(\cite{Gordon85})
\label{thm:Gordonneg1} Let $X_{ij}$ and $Y_{ij}$, $1\leq i\leq n,1\leq j\leq m$, be two centered Gaussian processes which satisfy the following inequalities for all choices of indices
\begin{enumerate}
\item $E(X_{ij}^2)=E(Y_{ij}^2)$
\item $E(X_{ij}X_{ik})\geq E(Y_{ij}Y_{ik})$
\item $E(X_{ij}X_{lk})\leq E(Y_{ij}Y_{lk}), i\neq l$.
\end{enumerate}
Let $\psi()$ be an increasing function on the real axis. Then
\begin{equation*}
E(\min_{i}\max_{j}\psi(X_{ij}))\leq E(\min_{i}\max_{j}\psi(Y_{ij})).
\end{equation*}
Moreover, let $\psi()$ be a decreasing function on the real axis. Then
\begin{equation*}
E(\max_{i}\min_{j}\psi(X_{ij}))\geq E(\max_{i}\min_{j}\psi(Y_{ij})).
\end{equation*}
\begin{proof}
The proof of all statements but the last one is of course given in \cite{Gordon85}. The proof of the last statement trivially follows and is given for completeness in \cite{StojnicMoreSophHopBnds10}.
\end{proof}
\end{theorem}

To make use of the above theorem we start by reformulating the problem in (\ref{eq:negham1}) in the following way
\begin{equation}
\xi_{sec}=\min_{\w\in\Ssec}\max_{\|\y\|_2=1}\y^TA\w.\label{eq:sqrtnegham2}
\end{equation}
In \cite{StojnicMoreSophHopBnds10} we established a lemma very similar to the following one:
\begin{lemma}
Let $A$ be an $m\times n$ matrix with i.i.d. standard normal components. Let $\g$ and $\h$ be $n\times 1$ and $m\times 1$ vectors, respectively, with i.i.d. standard normal components. Also, let $g$ be a standard normal random variable and let $c_3$ be a positive constant. Then
\begin{equation}
E(\max_{\w\in\Ssec}\min_{\|\y\|_2=1}e^{-c_3(\y^T A\w + g)})\leq E(\max_{\w\in\Ssec}\min_{\|\y\|_2=1}e^{-c_3(\g^T\y+\h^T\w)}).\label{eq:negexplemma}
\end{equation}\label{lemma:negexplemma}
\end{lemma}
\begin{proof}
As mentioned in \cite{StojnicMoreSophHopBnds10}, the proof is a standard/direct application of Theorem \ref{thm:Gordonneg1}. We will omit the details since they are pretty much the same as the those in the proof of the corresponding lemma in \cite{StojnicMoreSophHopBnds10}. However, we do mention that the only difference between this lemma and the one in \cite{StojnicMoreSophHopBnds10} is in set $\Ssec$. What is here $\Ssec$ it is a hypercube subset of $S^{n-1}$ in the corresponding lemma in \cite{StojnicMoreSophHopBnds10}. However, such a difference introduces no structural changes in the proof.
\end{proof}

Following step by step what was done after Lemma 3 in \cite{StojnicMoreSophHopBnds10} one arrives at the following analogue of \cite{StojnicMoreSophHopBnds10}'s equation $(57)$:
\begin{equation}
E(\min_{\w\in\Ssec}\|A\w\|_2)\geq
\frac{c_3}{2}-\frac{1}{c_3}\log(E(\max_{\w\in\Ssec}(e^{-c_3\h^T\w})))
-\frac{1}{c_3}\log(E(\min_{\|\y\|_2=1}(e^{-c_3\g^T\y}))).\label{eq:chneg8}
\end{equation}
Let $c_3=c_3^{(s)}\sqrt{n}$ where $c_3^{(s)}$ is a constant independent of $n$. Then (\ref{eq:chneg8}) becomes
\begin{eqnarray}
\hspace{-.5in}\frac{E(\min_{\w\in\Ssec}\|A\w\|_2)}{\sqrt{n}}
& \geq &
\frac{c_3^{(s)}}{2}-\frac{1}{nc_3^{(s)}}\log(E(\max_{\w\in\Ssec}(e^{-c_3^{(s)}\sqrt{n}\h^T\w})))
-\frac{1}{nc_3^{(s)}}\log(E(\min_{\|\y\|_2=1}(e^{-c_3^{(s)}\sqrt{n}\g^T\y})))\nonumber \\
& = &-(-\frac{c_3^{(s)}}{2}+I_{sec}(c_3^{(s)},\beta)+I_{sph}(c_3^{(s)},\alpha)),\label{eq:chneg9}
\end{eqnarray}
where
\begin{eqnarray}
I_{sec}(c_3^{(s)},\beta) & = & \frac{1}{nc_3^{(s)}}\log(E(\max_{\w\in\Ssec}(e^{-c_3^{(s)}\sqrt{n}\h^T\w})))\nonumber \\
I_{sph}(c_3^{(s)},\alpha) & = & \frac{1}{nc_3^{(s)}}\log(E(\min_{\|\y\|_2=1}(e^{-c_3^{(s)}\sqrt{n}\g^T\y}))).\label{eq:defIs}
\end{eqnarray}

One should now note that the above bound is effectively correct for any positive constant $c_3^{(s)}$. The only thing that is then left to be done so that the above bound becomes operational is to estimate $I_{sec}(c_3^{(s)},\beta)$ and $I_{sph}(c_3^{(s)},\alpha)$.

We start with $I_{sph}(c_3^{(s)},\alpha)$. In fact we just recall on the estimate that was provided in \cite{StojnicMoreSophHopBnds10}. Clearly, $I_{sph}(c_3^{(s)},\alpha)=E(\min_{\|\y\|_2=1}(e^{-c_3^{(s)}\sqrt{n}\g^T\y}))=Ee^{-c_3^{(s)}\sqrt{n}\|\g\|_2}$. As mentioned in \cite{StojnicMoreSophHopBnds10}, pretty good estimates for this quantity can be obtained for any $n$. However, to facilitate the exposition we will focus only on the large $n$ scenario (plus we only consider the concentrating scenario which pretty much implies that $n$ is large). In that case one can use the saddle point concept applied in \cite{SPH}. However, as in \cite{StojnicMoreSophHopBnds10}, here we will try to avoid the entire presentation from there and instead present the core neat idea that has much wider applications. Namely, we start with the following identity
\begin{equation}
-\|\g\|_2=\max_{\gamma_{sph}\geq 0}(-\frac{\|\g\|_2^2}{4\gamma_{sph}}-\gamma_{sph}).\label{eq:gamaiden}
\end{equation}
Then
\begin{multline}
\hspace{-.3in}\frac{1}{nc_3^{(s)}}\log(Ee^{-c_3^{(s)}\sqrt{n}\|\g\|_2})=\frac{1}{nc_3^{(s)}}\log(Ee^{c_3^{(s)}\sqrt{n}\max_{\gamma_{sph}\geq 0}(-\frac{\|\g\|_2^2}{4\gamma_{sph}}-\gamma_{sph})})
\doteq \frac{1}{nc_3^{(s)}}\max_{\gamma_{sph}\geq 0}\log(Ee^{-c_3^{(s)}\sqrt{n}(\frac{\|\g\|_2^2}{4\gamma_{sph}}+\gamma_{sph})})\\
=\max_{\gamma_{sph}\geq 0}(-\frac{\gamma_{sph}}{\sqrt{n}}+\frac{1}{c_3^{(s)}}\log(Ee^{-c_3^{(s)}\sqrt{n}(\frac{\g_i^2}{4\gamma_{sph}})}),\label{eq:gamaiden1}
\end{multline}
where $\doteq$ stands for equality when $n\rightarrow \infty$. $\doteq$ is among the results shown in \cite{SPH}. We do however, mention that one does not necessarily need $n\rightarrow\infty$ condition, i.e. the mechanism of \cite{SPH} can work with finite $n$ and provide the error terms; of course the writing is horrendously more complicated and we skip redoing it here (however, to emphasize this we avoided using explicitly the limit in (\ref{eq:gamaiden1})). Now if one sets $\gamma_{sph}=\gamma_{sph}^{(s)}\sqrt{n}$ then (\ref{eq:gamaiden1}) gives
\begin{multline}
\frac{1}{nc_3^{(s)}}\log(Ee^{-c_3^{(s)}\sqrt{n}\|\g\|_2})
=\max_{\gamma_{sph}^{(s)}\geq 0}(-\gamma_{sph}^{(s)}+\frac{1}{c_3^{(s)}}\log(Ee^{-c_3^{(s)}(\frac{\g_i^2}{4\gamma_{sph}^{(s)}})}))
=\max_{\gamma_{sph}^{(s)}\geq 0}(-\gamma_{sph}^{(s)}-\frac{\alpha}{2c_3^{(s)}}\log(1+\frac{c_3^{(s)}}{2\gamma_{sph}^{(s)}}))\\
=\max_{\gamma_{sph}^{(s)}\leq 0}(\gamma_{sph}^{(s)}-\frac{\alpha}{2c_3^{(s)}}\log(1-\frac{c_3^{(s)}}{2\gamma_{sph}^{(s)}})).\label{eq:gamaiden2}
\end{multline}
After solving the last maximization one obtains
\begin{equation}
\widehat{\gamma_{sph}^{(s)}}=\frac{2c_3^{(s)}-\sqrt{4(c_3^{(s)})^2+16\alpha}}{8}.\label{eq:gamaiden3}
\end{equation}
Connecting (\ref{eq:defIs}), (\ref{eq:gamaiden1}), (\ref{eq:gamaiden2}), and (\ref{eq:gamaiden3}) one finally has
\begin{equation}
I_{sph}(c_3^{(s)},\alpha)=\frac{1}{nc_3^{(s)}}\log(Ee^{-c_3^{(s)}\sqrt{n}\|\g\|_2})\doteq
\left ( \widehat{\gamma_{sph}^{(s)}}-\frac{\alpha}{2c_3^{(s)}}\log(1-\frac{c_3^{(s)}}{2\widehat{\gamma_{sph}^{(s)}}}\right ).\label{eq:Isph}
\end{equation}
where clearly $\widehat{\gamma_{sph}^{(s)}}$ is as in (\ref{eq:gamaiden3}).

We now switch to $I_{sec}(c_3^{(s)},\beta)$. Similarly to what was stated above, pretty good estimates for this quantity can be obtained for any $n$. However, to facilitate the exposition we will focus only on the large $n$ scenario. In that case one can again use the saddle point concept applied in \cite{SPH}. As above, we present the core idea without all the details from \cite{SPH}. Let $f(\w)=-\h^T\w$ and
we start with the following line of identities
\begin{multline}
\max_{\w\in\Ssec}f(\w)=-\min_{\w\in\Ssec}\h^T\w=-\min_{\w}\max_{\gamma_{sec}\geq 0,\nu_{sec}\geq 0} \h^T\w
-\nu_{sec}\sum_{i=n-k+1}^{n}|\w_i|
+\nu_{sec}\sum_{i=1}^{n-k}|\w_i|+\gamma_{sec}\sum_{i=1}^{n}\w_i^2-\gamma_{sec}\\
=-\max_{\gamma_{sec}\geq 0,\nu_{sec}\geq 0}\min_{\w} -\sum_{i=n-k+1}^{n}|\w_i|(|\h_i|+\nu_{sec})
+\sum_{i=1}^{n-k}|\w_i|(-|\h_i|+\nu_{sec})+\gamma_{sec}\sum_{i=1}^{n}\w_i^2-\gamma_{sec}\\
=-\max_{\gamma_{sec}\geq 0,\nu_{sec}\geq 0} -\frac{1}{4\gamma_{sec}}\left (\sum_{i=n-k+1}^{n}(|\h_i|+\nu_{sec})^2
+\sum_{i=1}^{n-k}\min(-|\h_i|+\nu_{sec},0)^2\right )-\gamma_{sec}\\
=\min_{\gamma_{sec}\geq 0,\nu_{sec}\geq 0} \frac{1}{4\gamma_{sec}}\left (\sum_{i=n-k+1}^{n}(|\h_i|+\nu_{sec})^2
+\sum_{i=1}^{n-k}\max(|\h_i|-\nu_{sec},0)^2\right )+\gamma_{sec}\\
=\min_{\gamma_{sec}\geq 0,\nu_{sec}\geq 0} \frac{f_1(\h,\nu_{sec},\beta)}{4\gamma_{sec}}+\gamma_{sec},\label{eq:seceq1}
\end{multline}
where
\begin{equation}
f_1(\h,\nu_{sec},\beta)=\left (\sum_{i=n-k+1}^{n}(|\h_i|+\nu_{sec})^2
+\sum_{i=1}^{n-k}\max(|\h_i|-\nu_{sec},0)^2\right ).\label{eq:deff1}
\end{equation}
Then
\begin{multline}
\hspace{-.3in}I_{sec}(c_3^{(s)},\beta)  =  \frac{1}{nc_3^{(s)}}\log(E(\max_{\w\in\Ssec}(e^{-c_3^{(s)}\sqrt{n}\h^T\w}))) = \frac{1}{nc_3^{(s)}}\log(E(\max_{\w\in\Ssec}(e^{c_3^{(s)}\sqrt{n}f(\w))})))\\=\frac{1}{nc_3^{(s)}}\log(Ee^{c_3^{(s)}\sqrt{n}\min_{\gamma_{sec},\nu_{sec}\geq 0}(\frac{f_1(\h,\nu_{sec},\beta)}{4\gamma_{sec}}+\gamma_{sec})})
\doteq \frac{1}{nc_3^{(s)}}\min_{\gamma_{sec},\nu_{sec}\geq 0}\log(Ee^{c_3^{(s)}\sqrt{n}(\frac{f_1(\h,\nu_{sec},\beta)}{4\gamma_{sec}}+\gamma_{sec})})\\
=\min_{\gamma_{sec},\nu_{sec}\geq 0}(\frac{\gamma_{sec}}{\sqrt{n}}+\frac{1}{nc_3^{(s)}}\log(Ee^{c_3^{(s)}\sqrt{n}(\frac{f_1(\h,\nu_{sec},\beta)}{4\gamma_{sec}})})),\label{eq:gamaiden1sec}
\end{multline}
where, as earlier, $\doteq$ stands for equality when $n\rightarrow \infty$ and would be obtained through the mechanism presented in \cite{SPH} (for our needs here though, even just replacing $\doteq$ with an $\leq$ inequality suffices). Now if one sets $\gamma_{sec}=\gamma_{sec}^{(s)}\sqrt{n}$ then (\ref{eq:gamaiden1}) gives
\begin{multline}
I_{sec}(c_3^{(s)},\beta)
=\min_{\gamma_{sec},\nu_{sec}\geq 0}(\frac{\gamma_{sec}}{\sqrt{n}}+\frac{1}{nc_3^{(s)}}\log(Ee^{c_3^{(s)}\sqrt{n}(\frac{f_1(\h,\nu_{sec},\beta)}{4\gamma_{sec}})})\\
=\min_{\gamma_{sec}^{(s)},\nu_{sec}\geq 0}(\gamma_{sec}^{(s)}+\frac{\beta}{c_3^{(s)}}\log(Ee^{(\frac{c_3^{(s)}(|\h_i|+\nu_{sec})^2}{4\gamma_{sec}^{(s)}})})
+\frac{1-\beta}{c_3^{(s)}}\log(Ee^{(\frac{c_3^{(s)}\max(|\h_i|-\nu_{sec},0)^2}{4\gamma_{sec}^{(s)}})}))\\
=\min_{\gamma_{sec}^{(s)},\nu_{sec}\geq 0}(\gamma_{sec}^{(s)}+\frac{\beta}{c_3^{(s)}}\log(I_{sec}^{(1)})
+\frac{1-\beta}{c_3^{(s)}}\log(I_{sec}^{(2)})),\label{eq:gamaiden2sec}
\end{multline}
where
\begin{eqnarray}
I_{sec}^{(1)} & = & Ee^{(\frac{c_3^{(s)}(|\h_i|+\nu_{sec})^2}{4\gamma_{sec}^{(s)}})}\nonumber \\
I_{sec}^{(2)} & = & Ee^{(\frac{c_3^{(s)}\max(|\h_i|-\nu_{sec},0)^2}{4\gamma_{sec}^{(s)}})}.\label{eq:defI1I2sec}
\end{eqnarray}
Now, to facilitate numerical computations we can create a bit more explicit expressions for the above quantities. We set $b=\frac{c_3^{(s)}}{4\gamma_{sec}^{(s)}}$ and obtain.
\begin{eqnarray}
I_{sec}^{(1)} & = & Ee^{(\frac{c_3^{(s)}(|\h_i|+\nu_{sec})^2}{4\gamma_{sec}^{(s)}})}
=\frac{e^{\frac{b\nu_{sec}^2}{1-2b}}}{\sqrt{1-2b}}(1+\mbox{erf}(\frac{\sqrt{2}b\nu_{sec}}{\sqrt{1-2b}}))\nonumber \\
I_{sec}^{(2)} & = & Ee^{(\frac{c_3^{(s)}\max(|\h_i|-\nu_{sec},0)^2}{4\gamma_{sec}^{(s)}})}
=\frac{e^{\frac{b\nu_{sec}^2}{1-2b}}}{\sqrt{1-2b}}(\mbox{erfc}(\frac{\nu_{sec}}{\sqrt{2(1-2b)}}))+\mbox{erf}(\frac{\nu_{sec}}{\sqrt{2}}),\label{eq:defI1I2sec1}
\end{eqnarray}
where of course to insure the integrals convergence we have $b<\frac{1}{2}$ or in other words $\gamma_{sec}^{(s)}>\frac{c_3^{(s)}}{2}$.

We summarize the above results related to the sectional threshold ($\beta_{sec}$) in the following theorem.

\begin{theorem}(Sectional threshold - lifted lower bound)
Let $A$ be an $m\times n$ measurement matrix in (\ref{eq:system})
with i.i.d. standard normal components. Let
the unknown $\x$ in (\ref{eq:system}) be $k$-sparse. Further, let the location of nonzero elements of $\x$ be arbitrarily chosen but fixed.
Let $k,m,n$ be large
and let $\alpha=\frac{m}{n}$ and $\betasec=\frac{k}{n}$ be constants
independent of $m$ and $n$. Let $\mbox{erf}$ be the standard error function associated with zero-mean unit variance Gaussian random variable and let $\mbox{erfc}=1-\mbox{erf}$.
Let
\begin{equation}
\widehat{\gamma_{sph}^{(s)}}=\frac{2c_3^{(s)}-\sqrt{4(c_3^{(s)})^2+16\alpha}}{8},\label{eq:gamasphthmsec}
\end{equation}
and
\begin{equation}
I_{sph}(c_3^{(s)},\alpha)=
\left ( \widehat{\gamma_{sph}^{(s)}}-\frac{\alpha}{2c_3^{(s)}}\log(1-\frac{c_3^{(s)}}{2\widehat{\gamma_{sph}^{(s)}}}\right ).\label{eq:Isphthmsec}
\end{equation}
Further, let $b=\frac{c_3^{(s)}}{4\gamma_{sec}^{(s)}}$,
\begin{eqnarray}
I_{sec}^{(1)}
& = & \frac{e^{\frac{b\nu_{sec}^2}{1-2b}}}{\sqrt{1-2b}}(1+\mbox{erf}(\frac{\sqrt{2}b\nu_{sec}}{\sqrt{1-2b}}))\nonumber \\
I_{sec}^{(2)} & = &
\frac{e^{\frac{b\nu_{sec}^2}{1-2b}}}{\sqrt{1-2b}}(\mbox{erfc}(\frac{\nu_{sec}}{\sqrt{2(1-2b)}}))+\mbox{erf}(\frac{\nu_{sec}}{\sqrt{2}}),\label{eq:defI1I2secthmsec}
\end{eqnarray}
and
\begin{equation}
I_{sec}(c_3^{(s)},\beta_{sec})=\min_{\gamma_{sec}^{(s)}\geq c_3^{(s)}/2,\nu_{sec}\geq 0}\left (\gamma_{sec}^{(s)}+\frac{\beta_{sec}}{c_3^{(s)}}\log(I_{sec}^{(1)})
+\frac{1-\beta_{sec}}{c_3^{(s)}}\log(I_{sec}^{(2)})\right ).\label{eq:Isecthmsec}
\end{equation}
If $\alpha$ and $\betasec$ are such that
\begin{equation}
\min_{c_3^{(s)}\geq 0}\left (-\frac{c_3^{(s)}}{2}+I_{sec}(c_3^{(s)},\beta_{sec})+I_{sph}(c_3^{(s)},\alpha)\right )<0,\label{eq:seccondthmsec}
\end{equation}
then the solution of (\ref{eq:l1}) is with overwhelming
probability the $k$-sparse $\x$ in (\ref{eq:system}).\label{thm:thmsecthrlift}
\end{theorem}
\begin{proof}
Follows from the above discussion.
\end{proof}

The results for the sectional threshold obtained from the above theorem as well as the corresponding ones from \cite{DonohoPol,DonohoUnsigned,StojnicCSetam09}
are presented in Figure \ref{fig:sec}.
\begin{figure}[htb]
\centering
\centerline{\epsfig{figure=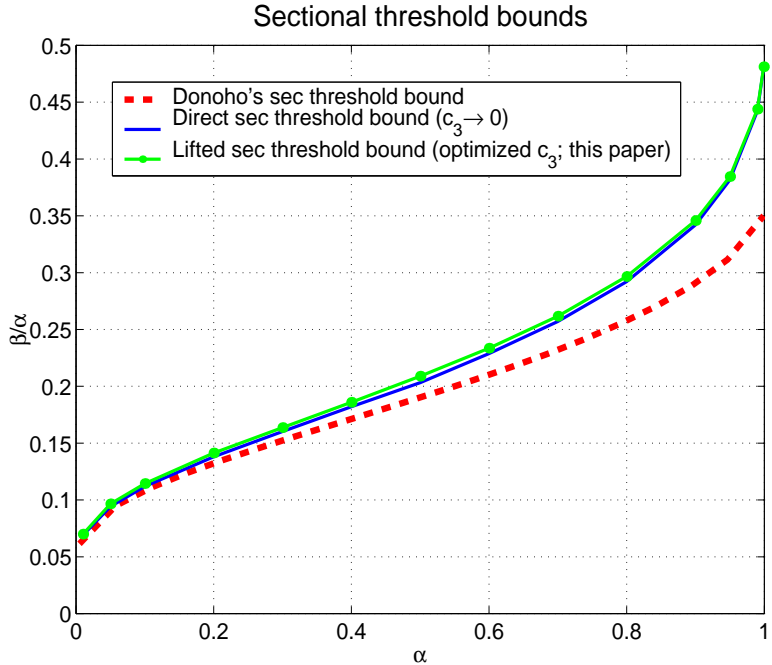,width=10.5cm,height=9cm}}
\caption{\emph{Sectional} threshold, $\ell_1$-optimization}
\label{fig:sec}
\end{figure}
The results slightly improve on those presented in \cite{StojnicCSetam09}. However, since the results are very close to the ones given in \cite{StojnicCSetam09} we present in Tables \ref{tab:sectab1} and \ref{tab:sectab2} the exact values of the attainable sectional thresholds obtained through the mechanisms presented in \cite{StojnicCSetam09} as well as those obtained through the mechanism presented in the above theorem. We also mention that the results presented in \cite{StojnicCSetam09} (and given in Theorem \ref{thm:thmsecthr}) can in fact be deduced from the above theorem. Namely, in the limit $c_3^{(s)}\rightarrow 0$, one from (\ref{eq:chneg9}) and (\ref{eq:defIs}) has $\max_{\w\in\Ssec}\h^T\w<\sqrt{\alpha n}$ as the limiting condition which is exactly the same condition considered in \cite{StojnicCSetam09}. For the completeness we present those results in the following corollary. (Of course we do mention that the results of the following corollary can be computed in a much faster fashion by simply analyzing $\max_{\w\in\Ssec}\h^T\w$ and realizing as in (\ref{eq:seceq1}) that $\max_{\w\in\Ssec}\h^T\w=\min_{\nu_{sec}}f_1(\h,\nu_{sec},\beta)$ where $f_1(\h,\nu_{sec},\beta)$ is as in (\ref{eq:deff1}); in fact this is exactly what was already done in \cite{StojnicCSetam09} and presented in Theorem \ref{thm:thmsecthr}. Here, our goal is rather different. Instead of handling this case directly we would like to show that it is in fact indeed a special case of the above theorem obtained for $c_3^{(s)}\rightarrow 0$.)

\begin{corollary}(Sectional threshold - lower bound)
Assume the setup of Theorem \ref{thm:thmsecthrlift}. Let $c_3^{(s)}\rightarrow 0$. Then
\begin{equation}
\widehat{\gamma_{sph}^{(s)}}\rightarrow -\frac{\sqrt{\alpha}}{2},\label{eq:gamasphcorsec}
\end{equation}
and
\begin{equation}
I_{sph}(c_3^{(s)},\alpha)\rightarrow
-\sqrt{\alpha}.\label{eq:Isphcorsec}
\end{equation}
Further, $b\rightarrow 0$ and
\begin{eqnarray}
I_{sec}^{(1)}
& \rightarrow &   \frac{e^{\frac{b\nu_{sec}^2}{1-2b}}}{\sqrt{1-2b}}(1+\frac{2}{\sqrt{\pi}}(\frac{\sqrt{2}b\nu_{sec}}{\sqrt{1-2b}}))
\rightarrow \frac{e^{\frac{b\nu_{sec}^2}{1-2b}}}{\sqrt{1-2b}}(1+\frac{2}{\sqrt{\pi}}(\sqrt{2}b\nu_{sec}))\nonumber \\
I_{sec}^{(2)} & \rightarrow & 1+
\mbox{erfc}(\frac{\nu_{sec}}{\sqrt{2(1-2b)}})(\frac{e^{\frac{b\nu_{sec}^2}{1-2b}}}{\sqrt{1-2b}}-1)+\mbox{erfc}(\frac{\nu_{sec}}{\sqrt{2(1-2b)}})
-\mbox{erfc}(\frac{\nu_{sec}}{\sqrt{2}})\nonumber \\
& \rightarrow &
1+
\mbox{erfc}(\frac{\nu_{sec}}{\sqrt{2}})(b\nu_{sec}^2+b)-\frac{2}{\sqrt{\pi}}\frac{\nu_{sec}}{\sqrt{2}}e^{-\frac{\nu_{sec}^2}{2}}b\nonumber \\
& \rightarrow & 1+
(\mbox{erfc}(\frac{\nu_{sec}}{\sqrt{2}})(1+\nu_{sec}^2)-\frac{2\nu_{sec}e^{-\frac{\nu_{sec}^2}{2}}}{\sqrt{2\pi}})b.\nonumber \\\label{eq:defI1I2seccorsec}
\end{eqnarray}
Moreover, let
\begin{eqnarray}
\hspace{-.4in}I_{sec}(c_3^{(s)},\beta_{sec}) & = & \min_{\gamma_{sec}^{(s)},\nu_{sec}\geq 0}
\left (\gamma_{sec}^{(s)}+\frac{\beta_{sec}(\nu_{sec}^2+1+2\sqrt{\frac{2}{\pi}}\nu_{sec})}{4\gamma_{sec}^{(s)}}+ \frac{(1-\beta_{sec})
(\mbox{erfc}(\frac{\nu_{sec}}{\sqrt{2}})(1+\nu_{sec}^2)-\frac{2\nu_{sec}e^{-\frac{\nu_{sec}^2}{2}}}{\sqrt{2\pi}})}{4\gamma_{sec}^{(s)}}\right )\nonumber \\
& = & \min_{\nu_{sec}\geq 0}
\sqrt{\left (\beta_{sec}(\nu_{sec}^2+1+2\sqrt{\frac{2}{\pi}}\nu_{sec})+ (1-\beta_{sec})
(\mbox{erfc}(\frac{\nu_{sec}}{\sqrt{2}})(1+\nu_{sec}^2)-\frac{2\nu_{sec}e^{-\frac{\nu_{sec}^2}{2}}}{\sqrt{2\pi}})\right )}.\label{eq:Iseccorsec}
\end{eqnarray}
If $\alpha$ and $\betasec$ are such that
\begin{eqnarray}
& & \min_{c_3^{(s)}\geq 0}\left (I_{sec}(c_3^{(s)},\beta_{sec})+I_{sph}(c_3^{(s)},\alpha)\right )<0\nonumber \\
&\Leftrightarrow &
\min_{\nu_{sec}\geq 0}
\sqrt{\left (\beta_{sec}(\nu_{sec}^2+1+2\sqrt{\frac{2}{\pi}}\nu_{sec})+ (1-\beta_{sec})
(\mbox{erfc}(\frac{\nu_{sec}}{\sqrt{2}})(1+\nu_{sec}^2)-\frac{2\nu_{sec}e^{-\frac{\nu_{sec}^2}{2}}}{\sqrt{2\pi}})\right )}<\sqrt{\alpha},\nonumber\\
\label{eq:seccondcorsec}
\end{eqnarray}
then the solution of (\ref{eq:l1}) is with overwhelming
probability the $k$-sparse $\x$ in (\ref{eq:system}).\label{thm:corsecthrlift}
\end{corollary}
\begin{proof}
Theorem \ref{thm:thmsecthrlift} holds for any $c_3^{(s)}\geq 0$. The above corollary instead of looking for the best possible $c_3^{(s)}$ in Theorem \ref{thm:thmsecthrlift} assumes a simple $c_3^{(s)}\rightarrow 0$ scenario.

Alternatively, one can look at $\frac{E\max_{\w\in\Ssec}\h^T\w}{\sqrt{n}}$ and following the methodology presented in (\ref{eq:seceq1}) (and originally in \cite{StojnicCSetam09}) through a combination of a combination of (\ref{eq:gamaiden2sec}), (\ref{eq:defI1I2sec}), (\ref{eq:Isecthmsec}), (\ref{eq:seccondthmsec}), and (\ref{eq:Isphcorsec}) obtain for a scalar $\nu_{sec}\geq 0$
\begin{eqnarray}
\frac{E\max_{\w\in\Ssec}\h^T\w}{\sqrt{n}}& \leq & \gamma_{sec}+\frac{\beta_{sec}E(|\h_i|+\nu_{sec})^2
+(1-\beta_{sec})E((\max(|\h_i|-\nu_{sec},0))^2)}{4\gamma_{sec}}\nonumber \\
& = & \sqrt{\beta_{sec}E(|\h_i|+\nu_{sec})^2+ (1-\beta_{sec})E_{\nu_{sec}\leq |\h_i|}(|\h_i|-\nu_{sec})^2}.\label{eq:corrsec0}
\end{eqnarray}
Optimizing (tightening) over $\nu_{sec}\geq 0$ (and using all the concentrating machinery of \cite{StojnicCSetam09}) one can write
\begin{eqnarray}
\hspace{-.5in}\frac{E\max_{\w\in\Ssec}\h^T\w}{\sqrt{n}} & \doteq & \min_{\nu_{sec}\geq 0}\sqrt{\left (\beta_{sec}(1+2\sqrt{\frac{2}{\pi}}\nu_{sec}+\nu_{sec}^2)
+(1-\beta_{sec})\int_{\nu_{sec}^{(1)}\leq |\h_i|}(|\h_i|-\nu_{sec})^2\frac{e^{-\frac{\h_i^2}{2}}d\h_i}{\sqrt{2\pi}}\right )}\nonumber \\
& = & \min_{\nu_{sec}\geq 0}\sqrt{\left (\beta_{sec}(1+2\sqrt{\frac{2}{\pi}}\nu_{sec}+\nu_{sec}^2)
+(1-\beta_{sec})(\mbox{erfc}(\frac{\nu_{sec}}{\sqrt{2}})(1+\nu_{sec}^2)-\frac{2\nu_{sec}e^{-\frac{\nu_{sec}^2}{2}}}{\sqrt{2\pi}})\right )}.\nonumber \\\label{eq:corrsec1}
\end{eqnarray}
Connecting beginning and end in (\ref{eq:corrsec1}) then leads to the condition given in the above corollary.
\end{proof}

\noindent \textbf{Remark:} Solving over $\nu_{sec}$ and juggling a bit one can arrive at the results presented in Theorem \ref{thm:thmsecthr}. Of course, the results of Theorem \ref{thm:thmsecthr} were presented in more detail in \cite{StojnicCSetam09} where our main concern was a thorough studying of the underlying optimization problem.

The results obtained from the previous corollary (and obviously the results from Theorem \ref{thm:thmsecthr}) are those presented in Figure \ref{fig:sec} and Tables \ref{tab:sectab1} and \ref{tab:sectab2} that we refer to as $c_3\rightarrow 0$ scenario or direct sectional threshold bounds (also we remove superscript $(s)$ from $c_3^{(s)}$ to make the figure and tables easier to view). As can be seen from the tables, while conceptually substantial in practice the improvement may not be fairly visible. That can be because the methods are not powerful enough to make a bigger improvement or simply because a big improvement may not be possible (in other words the results obtained in \cite{StojnicCSetam09} may very well already be fairly close to the optimal ones). As for the limits of the developed methods, we do want to emphasize that we did solve the numerical optimizations that appear in Theorem \ref{thm:thmsecthrlift} only on a local optimum level and obviously only with a finite precision. We do not know if a substantial change would occur in the presented results had we solved it on a global optimum level (we recall that finding local optima is of course certainly enough to establish attainable values of the sectional thresholds). As for how far away from the true sectional thresholds are the results presented in the tables and Figure \ref{fig:sec}, we actually believe that they are in fact very, very close to the optimal ones.

\begin{table}
\caption{Sectional threshold bounds -- low $\alpha\leq 0.5$ regime}\vspace{.1in}
\hspace{-0in}\centering
\begin{tabular}{||c|c|c|c|c|c|c|c||}\hline\hline
 $\alpha$  & $0.01$ & $0.05$ & $0.1$ & $0.2$ & $0.3$ & $0.4$ & $0.5$ \\ \hline\hline
 $\beta_{sec}$ ($c_3\rightarrow 0$) & $0.00069$ & $0.00471$ & $0.0112$ & $0.0276$ & $0.0481$ & $0.0728$ & $0.1022$ \\ \hline
 $\beta_{sec}$ (optimized $c_3$)
 &  $0.00070$ & $0.00483$ & $0.0115$ & $0.0283$ & $0.0491$ & $0.0744$ & $0.1045$ \\ \hline\hline
\end{tabular}
\label{tab:sectab1}
\end{table}

\begin{table}
\caption{Sectional threshold bounds -- high $\alpha> 0.5$ regime}\vspace{.1in}
\hspace{-0in}\centering
\begin{tabular}{||c|c|c|c|c|c|c|c|c||}\hline\hline
 $\alpha$  & $0.6$ & $0.7$ & $0.8$ & $0.9$ & $0.95$ & $0.99$ & $0.999$ & $0.9999$\\ \hline\hline
 $\beta_{sec}$ ($c_3\rightarrow 0$) & $0.1373$ & $0.1800$ & $0.2337$ &
$0.3079$ & $0.3626$  & $0.4378$ & $0.4802$ & $0.4937 $ \\ \hline
 $\beta_{sec}$ (optimized $c_3$)
 & $0.1401$ & $0.1832$ &
$0.2373$ & $0.3113$ & $0.3654$ & $0.4394$ &  $0.4807$ & $0.4937$ \\ \hline\hline
\end{tabular}
\label{tab:sectab2}
\end{table}

\section{Lifting $\ell_1$-minimization strong threshold}
\label{sec:strthr}

In this section we look at the strong thresholds. We do mention before presenting anything that these strong threshold results were substantially harder to establish than the ones that we presented in the previous subsection for the sectional thresholds.

As in the previous sections, throughout the presentation in this section we will assume a substantial level of familiarity with many of the well-known results that relate to the performance characterization of (\ref{eq:l1}) (we will again fairly often recall on many results/definitions that we established in \cite{StojnicCSetam09}). We start by defining a set $\Sstr$
\begin{equation}
\Sstr=\{\w\in S^{n-1}| \quad \sum_{i=1}^n \b_i|\w_i|<0,\b_i^2=1,\sum_{i=1}^n\b_i=n-2k\},\label{eq:defSstr}
\end{equation}
where $S^{n-1}$ is the unit sphere in $R^n$. Then it was established in \cite{StojnicCSetam09} that the following optimization problem is of critical importance in determining the strong threshold of $\ell_1$-minimization
\begin{equation}
\xi_{str}=\min_{\w\in\Sstr}\|A\w\|_2.\label{eq:negham1str}
\end{equation}
Namely, what was established in \cite{StojnicCSetam09} is roughly the following: if $\xi_{str}$ is positive with overwhelming probability for certain combination of $k$, $m$, and $n$ then for $\alpha=\frac{m}{n}$ one has a lower bound $\beta_{str}=\frac{k}{n}$ on the true value of the strong threshold with overwhelming probability. Also, as was the case with the sectional thresholds studied in the previous section, the mechanisms of \cite{StojnicCSetam09} were powerful enough to establish the concentration of $\xi_{str}$ as well. This essentially means that if we can show that $E\xi_{str}>0$ for certain $k$, $m$, and $n$ we can then obtain the lower bound on the strong threshold. This is, of course, precisely what was done in \cite{StojnicCSetam09}. However, as was the case with the sectional thresholds, the results we obtained for the strong threshold in \cite{StojnicCSetam09} through such a consideration were not exact. The main reason of course was inability to determine $E\xi_{str}$ exactly. Instead we resorted to its lower bounds and similarly to what happened when we consider sectional thresholds, those bounds turned out to be loose. Moreover, they were loose enough that in certain range of $\alpha$-axis the obtained thresholds could not even achieve the lower bounds already known from \cite{DonohoPol}.  In this section we will use some of the ideas from the previous section (which have roots in the mechanisms we recently introduced in \cite{StojnicMoreSophHopBnds10}) to provide a substantial conceptual improvement in these bounds. This would in turn reflect in a conceptual improvement of the strong thresholds. However, we do mention that although the translation from the results achieved in the previous section to the ones that we will present below will appear seemingly smooth it was not so trivial to achieve such a translation. In fact, it took a substantial effort on our part to create results from the previous section and way, way more than that to find a good mechanism that can fit the strong thresholds as well.

Below we present a way to create a lower-bound on the optimal value of (\ref{eq:negham1str}).

\subsection{Lower-bounding $\xi_{str}$}
\label{sec:lbxistr}

As mentioned above, in this subsection we will look at problem from (\ref{eq:negham1str}) or more precisely its optimal value. Also, as mentioned earlier, we will continue to assume that the elements of $A$ are i.i.d. standard normal random variables.

We start by reformulating the problem in (\ref{eq:negham1str}) in the following way
\begin{equation}
\xi_{str}=\min_{\w\in\Sstr}\max_{\|\y\|_2=1}\y^TA\w.\label{eq:sqrtnegham2str}
\end{equation}
Then we continue by reformulating Lemma \ref{lemma:negexplemma}. Namely, based on Lemma \ref{lemma:negexplemma}, we establish the following:
\begin{lemma}
Let $A$ be an $m\times n$ matrix with i.i.d. standard normal components. Let $\g$ and $\h$ be $n\times 1$ and $m\times 1$ vectors, respectively, with i.i.d. standard normal components. Also, let $g$ be a standard normal random variable and let $c_3$ be a positive constant. Then
\begin{equation}
E(\max_{\w\in\Sstr}\min_{\|\y\|_2=1}e^{-c_3(\y^T A\w + g)})\leq E(\max_{\w\in\Sstr}\min_{\|\y\|_2=1}e^{-c_3(\g^T\y+\h^T\w)}).\label{eq:negexplemma}
\end{equation}\label{lemma:negexplemmastr}
\end{lemma}
\begin{proof}
The proof is a standard/direct application of Theorem \ref{thm:Gordonneg1} (as was the proof of Lemma \ref{lemma:negexplemma}). We will of course omit the details  an just mention that the only difference between this lemma and Lemma \ref{lemma:negexplemma} is in the structure of set $\Sstr$. What is here $\Sstr$ it is $\Ssec$ in Lemma \ref{lemma:negexplemma}. However, such a difference introduces no structural changes in the proof.
\end{proof}

Following what was done in the previous section we arrive at the following analogue of (\ref{eq:chneg8}) (and ultimately of \cite{StojnicMoreSophHopBnds10}'s equation $(57)$):
\begin{equation}
E(\min_{\w\in\Sstr}\|A\w\|_2)\geq
\frac{c_3}{2}-\frac{1}{c_3}\log(E(\max_{\w\in\Sstr}(e^{-c_3\h^T\w})))
-\frac{1}{c_3}\log(E(\min_{\|\y\|_2=1}(e^{-c_3\g^T\y}))).\label{eq:chneg8str}
\end{equation}
As earlier, let $c_3=c_3^{(s)}\sqrt{n}$ where $c_3^{(s)}$ is a constant independent of $n$. Then (\ref{eq:chneg8str}) becomes
\begin{eqnarray}
\hspace{-.5in}\frac{E(\min_{\w\in\Sstr}\|A\w\|_2)}{\sqrt{n}}
& \geq &
\frac{c_3^{(s)}}{2}-\frac{1}{nc_3^{(s)}}\log(E(\max_{\w\in\Sstr}(e^{-c_3^{(s)}\sqrt{n}\h^T\w})))
-\frac{1}{nc_3^{(s)}}\log(E(\min_{\|\y\|_2=1}(e^{-c_3^{(s)}\sqrt{n}\g^T\y})))\nonumber \\
& = &-(-\frac{c_3^{(s)}}{2}+I_{str}(c_3^{(s)},\beta)+I_{sph}(c_3^{(s)},\alpha)),\label{eq:chneg9str}
\end{eqnarray}
where
\begin{eqnarray}
I_{str}(c_3^{(s)},\beta) & = & \frac{1}{nc_3^{(s)}}\log(E(\max_{\w\in\Sstr}(e^{-c_3^{(s)}\sqrt{n}\h^T\w})))\nonumber \\
I_{sph}(c_3^{(s)},\alpha) & = & \frac{1}{nc_3^{(s)}}\log(E(\min_{\|\y\|_2=1}(e^{-c_3^{(s)}\sqrt{n}\g^T\y}))).\label{eq:defIsstr}
\end{eqnarray}

As in the previous section, one should now note that the above bound is effectively correct for any positive constant $c_3^{(s)}$. The only thing that is then left to be done so that the above bound becomes operational is to estimate $I_{str}(c_3^{(s)},\beta)$ and $I_{sph}(c_3^{(s)},\alpha)$. We recall that
\begin{equation}
I_{sph}(c_3^{(s)},\alpha)=\frac{1}{nc_3^{(s)}}\log(Ee^{-c_3^{(s)}\sqrt{n}\|\g\|_2})\doteq
\left ( \widehat{\gamma_{sph}^{(s)}}-\frac{\alpha}{2c_3^{(s)}}\log(1-\frac{c_3^{(s)}}{2\widehat{\gamma_{sph}^{(s)}}}\right ),\label{eq:Isphstr}
\end{equation}
where
\begin{equation}
\widehat{\gamma_{sph}^{(s)}}=\frac{2c_3^{(s)}-\sqrt{4(c_3^{(s)})^2+16\alpha}}{8}.\label{eq:gamaiden3str}
\end{equation}

We now switch to $I_{str}(c_3^{(s)},\beta)$. Similarly to what was stated earlier, pretty good estimates for this quantity can be obtained for any $n$. However, to facilitate the exposition we will focus only on the large $n$ scenario. In that case one can again use the saddle point concept applied in \cite{SPH}. As above, we present the core idea without all the details from \cite{SPH}. Let $f(\w)=-\h^T\w$ and
we start with the following line of identities
\begin{eqnarray}
f_{str}=\max_{\w\in\Sstr}f(\w)=-\min_{\w\in\Sstr}\h^T\w =  -\min_{\b,\w} & & \h^T\w\nonumber \\
\mbox{subject to} & & \sum_{i=1}^{n} \b_i|\w_i|<0,\nonumber \\
& & \sum_{i=1}^{n} \w_i^2=1,\nonumber \\
& & \b_i\in\{-1,1\},1 i\leq n,\nonumber \\
& & \sum_{i=1}^{n}\b_i= n-2k.\label{eq:streq0}
\end{eqnarray}
We then further have
\begin{eqnarray}
f_{str} & = & -\min_{\b_i^2=1,\w}\max_{\gamma_{str},\nu_{str}^{(1)},\nu_{str}^{(2)}\geq 0} \h^T\w+\nu_{str}^{(1)}\sum_{i=1}^{n} \b_i|\w_i|
-\nu_{str}^{(2)}\sum_{i=1}^{n}\b_i+\nu_{str}^{(2)}(n-2k)+\gamma_{str}\sum_{i=1}^{n} \w_i^2-\gamma_{str}\nonumber \\
&\leq & -\max_{\gamma_{str},\nu_{str}^{(1)},\nu_{str}^{(2)}\geq 0}\min_{\b_i^2=1,\w} \h^T\w+\nu_{str}^{(1)}\sum_{i=1}^{n} \b_i|\w_i|
-\nu_{str}^{(2)}\sum_{i=1}^{n}\b_i+\nu_{str}^{(2)}(n-2k)+\gamma_{str}\sum_{i=1}^{n} \w_i^2-\gamma_{str}\nonumber \\
& = & -\max_{\gamma_{str},\nu_{str}^{(1)},\nu_{str}^{(2)}\geq 0}\min_{\b_i^2=1,\w} \sum_{i=1}^{n} (-|\h_i|+\nu_{str}^{(1)}\b_i)|\w_i|
-\nu_{str}^{(2)}\sum_{i=1}^{n}\b_i+\nu_{str}^{(2)}(n-2k)+\gamma_{str}\sum_{i=1}^{n} \w_i^2-\gamma_{str}.\nonumber \\
\label{eq:streq01}
\end{eqnarray}
Positivity condition on $\nu_{str}^{(2)}$ is added although it is not necessary (it essentially amount to relaxing the last constraint to an inequality which changes nothing with respect to the final results). Optimizing further we obtain
\begin{eqnarray}
f_{str} & \leq & -\max_{\gamma_{str},\nu_{str}^{(1)},\nu_{str}^{(2)}\geq 0}\min_{\b_i^2=1,\w} \sum_{i=1}^{n} (-|\h_i|+\nu_{str}^{(1)}\b_i)|\w_i|
-\nu_{str}^{(2)}\sum_{i=1}^{n}\b_i+\nu_{str}^{(2)}(n-2k)+\gamma_{str}\sum_{i=1}^{n} \w_i^2-\gamma_{str}\nonumber \\
& = & \min_{\gamma_{str},\nu_{str}^{(1)},\nu_{str}^{(2)}\geq 0}\max_{\b_i^2=1,\w} \sum_{i=1}^{n} (|\h_i|-\nu_{str}^{(1)}\b_i)|\w_i|
+\nu_{str}^{(2)}\sum_{i=1}^{n}\b_i-\nu_{str}^{(2)}(n-2k)-\gamma_{str}\sum_{i=1}^{n} \w_i^2+\gamma_{str}.\nonumber \\
\label{eq:streq02}
\end{eqnarray}
To solve the inner optimization it helps to introduce a vector $\t$ in the following way
\begin{equation}
\t_i=\max\left (\frac{(|\h_i|+\nu_{str}^{(1)})^2}{4\gamma_{str}}-\nu_{str}^{(2)},
\frac{(\max(|\h_i|-\nu_{str}^{(1)},0))^2}{4\gamma_{str}}+\nu_{str}^{(2)}\right ),\label{eq:streq030}
\end{equation}
or alternatively (possibly in a more convenient way)
\begin{equation}
\t_i=\begin{cases}\frac{\h_i^2+(\nu_{str}^{(1)})^2}{4\gamma_{str}}+|\frac{|\h_i|\nu_{str}^{(1)}}{2\gamma_{str}}-\nu_{str}^{(2)}|, & |\h_i|\geq \nu_{str}^{(1)}\\
\max\left (\frac{(|\h_i|+\nu_{str}^{(1)})^2}{4\gamma_{str}}-\nu_{str}^{(2)},\nu_{str}^{(2)}\right ), &  |\h_i|\leq \nu_{str}^{(1)}.\label{eq:streq03}
\end{cases}
\end{equation}
Using (\ref{eq:streq03}), (\ref{eq:streq02}) then becomes
\begin{eqnarray}
f_{str} & \leq &  \min_{\gamma_{str},\nu_{str}^{(1)},\nu_{str}^{(2)}\geq 0}\max_{\b_i^2=1,\w} \sum_{i=1}^{n} (|\h_i|-\nu_{str}^{(1)}\b_i)|\w_i|
+\nu_{str}^{(2)}\sum_{i=1}^{n}\b_i-\nu_{str}^{(2)}(n-2k)-\gamma_{str}\sum_{i=1}^{n} \w_i^2+\gamma_{str}\nonumber \\
& = &  \min_{\gamma_{str},\nu_{str}^{(1)},\nu_{str}^{(2)}\geq 0} \sum_{i=1}^{n}\t_i+\nu_{str}^{(2)}(2k-n)+\gamma_{str}.\nonumber \\
\label{eq:streq04}
\end{eqnarray}
Although we showed an inequality on $f_{str}$ (which is sufficient for what we need here) we do mention that the above actually holds with the equality. Let
\begin{equation}
f_1^{(str)}(\h,\nu_{str}^{(1)},\nu_{str}^{(2)},\gamma_{str},\beta)=\sum_{i=1}^{n} \t_i.\label{eq:deff1str}
\end{equation}
Then
\begin{multline}
\hspace{-.3in}I_{str}(c_3^{(s)},\beta)  =  \frac{1}{nc_3^{(s)}}\log(E(\max_{\w\in\Sstr}(e^{-c_3^{(s)}\sqrt{n}\h^T\w}))) = \frac{1}{nc_3^{(s)}}\log(E(\max_{\w\in\Sstr}(e^{c_3^{(s)}\sqrt{n}f(\w))})))\\=\frac{1}{nc_3^{(s)}}\log(Ee^{c_3^{(s)}\sqrt{n}\min_{\gamma_{str},\nu_{str}^{(1)}
,\nu_{str}^{(2)}\geq 0}(f_1^{(str)}(\h,\nu_{str}^{(1)},\nu_{str}^{(2)},\gamma_{str},\beta)+\nu_{str}^{(2)}(2k-n)+\gamma_{str})})\\
\doteq \frac{1}{nc_3^{(s)}}\min_{\gamma_{str},\nu_{str}^{(1)}
,\nu_{str}^{(2)}\geq 0}\log(Ee^{c_3^{(s)}\sqrt{n}(f_1^{(str)}(\h,\nu_{str}^{(1)},\nu_{str}^{(2)},\gamma_{str},\beta)+\nu_{str}^{(2)}(2k-n)+\gamma_{str})})\\
=\min_{\gamma_{str},\nu_{str}^{(1)}
,\nu_{str}^{(2)}\geq 0}(\nu_{str}^{(2)}\sqrt{n}(2\beta-1)+ \frac{\gamma_{str}}{\sqrt{n}}+\frac{1}{nc_3^{(s)}}\log(Ee^{c_3^{(s)}\sqrt{n}(f_1^{(str)}(\h,\nu_{str}^{(1)},\nu_{str}^{(2)},\gamma_{str},\beta))}))\\
=\min_{\gamma_{str},\nu_{str}^{(1)}
,\nu_{str}^{(2)}\geq 0}(\nu_{str}^{(2)}\sqrt{n}(2\beta-1)+ \frac{\gamma_{str}}{\sqrt{n}}+\frac{1}{nc_3^{(s)}}\log(Ee^{c_3^{(s)}\sqrt{n}(\sum_{i=1}^{n}\t_i)})),\label{eq:gamaiden1str}
\end{multline}
where $\t_i$ is as given in (\ref{eq:streq03}) and as earlier, $\doteq$ stands for equality when $n\rightarrow \infty$ and would be obtained through the mechanism presented in \cite{SPH} (for our needs here though, even just replacing $\doteq$ with a simple $\leq$ inequality suffices). Now if one sets $\gamma_{str}=\gamma_{str}^{(s)}\sqrt{n}$ and $\nu_{str}^{(2,s)}=\nu_{str}^{(2)}\sqrt{n}$ then (\ref{eq:gamaiden1str}) gives
\begin{eqnarray}
I_{str}(c_3^{(s)},\beta)
& = & \min_{\gamma_{str},\nu_{str}^{(1)}
,\nu_{str}^{(2)}\geq 0}(\nu_{str}^{(2)}\sqrt{n}(2\beta-1)+ \frac{\gamma_{str}}{\sqrt{n}}+\frac{1}{nc_3^{(s)}}\log(Ee^{c_3^{(s)}\sqrt{n}(\sum_{i=1}^{n}\t_i)}))\nonumber \\
& = &
\min_{\gamma_{str}^{(s)},\nu_{str}^{(1)}
,\nu_{str}^{(2,s)}\geq 0}(\nu_{str}^{(2,s)}(2\beta-1)+ \gamma_{str}^{(s)}+\frac{1}{c_3^{(s)}}\log(Ee^{c_3^{(s)}\t_i^{(s)}})),\label{eq:gamaiden2str}
\end{eqnarray}
where
\begin{equation}
\t_i^{(s)}=\begin{cases}\frac{\h_i^2+(\nu_{str}^{(1)})^2}{4\gamma_{str}^{(s)}}+|\frac{|\h_i|\nu_{str}^{(1)}}{2\gamma_{str}^{(s)}}-\nu_{str}^{(2,s)}|, & |\h_i|\geq \nu_{str}^{(1)}\\
\max\left (\frac{(|\h_i|+\nu_{str}^{(1)})^2}{4\gamma_{str}^{(s)}}-\nu_{str}^{(2,s)},\nu_{str}^{(2,s)}\right ), &  |\h_i|\leq \nu_{str}^{(1)}.\label{eq:deftisstr}
\end{cases}
\end{equation}
The above characterization is then sufficient to compute attainable strong thresholds. However, since there will be a substantial numerical work involved it is probably a bit more convenient to look for a neater representation. That obviously involves solving a bunch of integrals. We skip such a tedious jobs but present the final results. We start with setting
\begin{equation}
I=Ee^{c_3^{(s)}\t_i^{(s)}}.\label{eq:defbigIstr}
\end{equation}
Then one has
\begin{equation}
I=\begin{cases}I^{(1)}, & (\nu_{str}^{(1)})^2<2\gamma_{str}^{(s)}\nu_{str}^{(2,s)}\\
I^{(2)}, & 2\gamma_{str}^{(s)}\nu_{str}^{(2,s)}\leq (\nu_{str}^{(1)})^2<8\gamma_{str}^{(s)}\nu_{str}^{(2,s)}\\
I^{(3)}, & (\nu_{str}^{(1)})^2>8\gamma_{str}^{(s)}\nu_{str}^{(2,s)},\label{eq:defbigIstr1}
\end{cases}
\end{equation}
where $I^{(1)}$, $I^{(2)}$, and $I^{(3)}$ are defined below.

\underline{\emph{1. Determining $I^{(1)}$}}

\noindent Set
\begin{eqnarray}
p & = & c_3^{(s)}/4/\gamma_{str}^{(s)}\nonumber \\
q & = & c_3^{(s)}\nu_{str}^{(1)}/2/\gamma_{str}^{(s)}\nonumber \\
r & = & c_3^{(s)}((\nu_{str}^{(1)})^2/4/\gamma_{str}^{(s)}-\nu_{str}^{(2,s)})\nonumber \\
C_1 & = & exp((q/\sqrt{2(1-2p)})^2+r)/\sqrt{1/2-p}\nonumber \\
I_{11}^{(1)} & = & C_1/2/\sqrt{2}\mbox{erfc}(2\gamma_{str}^{(s)}\nu_{str}^{(2,s)}/\nu_{str}^{(1)}\sqrt{1/2-p}-q/\sqrt{2(1-2p)}),\label{eq:defbigI111str}
\end{eqnarray}
and
\begin{eqnarray}
r_1 & = & c_3^{(s)}((\nu_{str}^{(1)})^2/4/\gamma_{str}^{(s)}+\nu_{str}^{(2,s)})\nonumber \\
C_{12} & = & exp((-q/\sqrt{2(1-2p)})^2+r_1)/\sqrt{1/2-p}\nonumber \\
I_{12}^{(1)} & = & C_{12}/2/\sqrt{2}(\mbox{erfc}(\nu_{str}^{(1)}\sqrt{1/2-p}+q/\sqrt{2(1-2p)})
-\mbox{erfc}(2\gamma_{str}^{(s)}\nu_{str}^{(2,s)}/\nu_{str}^{(1)}\sqrt{1/2-p}+q/\sqrt{2(1-2p)})).\nonumber \\\label{eq:defbigI112str}
\end{eqnarray}
Further set
\begin{equation}
I_2^{(1)}=1/2e^{c_3^{(s)}\nu_{str}^{(2,s)}}\mbox{erf}(\nu_{str}^{(s)}/\sqrt{2}).\label{eq:defbigI12str}
\end{equation}
Then
\begin{equation}
I^{(1)}=2(I_{11}^{(1)}+I_{12}^{(1)} +I_2^{(1)}).\label{eq:fbigI1str}
\end{equation}

\underline{\emph{2. Determining $I^{(2)}$}}

\noindent Set
\begin{eqnarray}
I_1^{(2)} & = & C_1/2/\sqrt{2}\mbox{erfc}(\nu_{str}^{(1)}\sqrt{1/2-p}-q/\sqrt{2(1-2p)})\nonumber \\
I_{21}^{(2)} & = & 1/2e^{c_3^{(s)}\nu_{str}^{(2,s)}}\mbox{erf}((\sqrt{8\gamma_{str}^{(s)}\nu_{str}^{(2,s)}}-\nu_{str}^{(1)})/\sqrt{2})\nonumber \\
I_{22}^{(2)} & = & C_1/2/\sqrt{2}(\mbox{erfc}((\sqrt{8\gamma_{str}^{(s)}\nu_{str}^{(2,s)}}-\nu_{str}^{(1)})\sqrt{1/2-p}-q/\sqrt{2(1-2p)})\nonumber \\
& - & \mbox{erfc}(\nu_{str}^{(1)}\sqrt{1/2-p}-q/\sqrt{2(1-2p)})).
\label{eq:defbigI21str}
\end{eqnarray}
Then
\begin{equation}
I^{(2)}=2(I_{1}^{(2)}+I_{21}^{(1)} +I_{22}^{(2)}).\label{eq:fbigI2str}
\end{equation}

\underline{\emph{3. Determining $I^{(3)}$}}

\noindent Set
\begin{eqnarray}
I_{22}^{(3)}=C_1/2/\sqrt{2}(\mbox{erfc}(-q/\sqrt{2(1-2p)})-\mbox{erfc}(\nu_{str}^{(1)}\sqrt{1/2-p}-q/\sqrt{2(1-2p)})).
\label{eq:defbigI322str}
\end{eqnarray}
Then
\begin{equation}
I^{(3)}=2(I_{1}^{(2)}+I_{21}^{(1)} +I_{22}^{(3)}).\label{eq:fbigI3str}
\end{equation}

We summarize the above results related to the strong threshold ($\beta_{str}$) in the following theorem.

\begin{theorem}(Strong threshold - lifted lower bound)
Let $A$ be an $m\times n$ measurement matrix in (\ref{eq:system})
with i.i.d. standard normal components. Let
the unknown $\x$ in (\ref{eq:system}) be $k$-sparse.
Let $k,m,n$ be large
and let $\alpha=\frac{m}{n}$ and $\betastr=\frac{k}{n}$ be constants
independent of $m$ and $n$. Let $\mbox{erf}$ be the standard error function associated with zero-mean unit variance Gaussian random variable and let $\mbox{erfc}=1-\mbox{erf}$.
Let
\begin{equation}
\widehat{\gamma_{sph}^{(s)}}=\frac{2c_3^{(s)}-\sqrt{4(c_3^{(s)})^2+16\alpha}}{8},\label{eq:gamasphthmstr}
\end{equation}
and
\begin{equation}
I_{sph}(c_3^{(s)},\alpha)=
\left ( \widehat{\gamma_{sph}^{(s)}}-\frac{\alpha}{2c_3^{(s)}}\log(1-\frac{c_3^{(s)}}{2\widehat{\gamma_{sph}^{(s)}}}\right ).\label{eq:Isphthmstr}
\end{equation}
Further, let $I$ be defined through (\ref{eq:defbigIstr1})-(\ref{eq:fbigI3str}) (or alternatively through (\ref{eq:deftisstr}) and (\ref{eq:defbigIstr})) and
\begin{equation}
I_{str}(c_3^{(s)},\beta_{str})=\min_{\gamma_{str}^{(s)}\geq c_3^{(s)}/2,\nu_{str}^{(1)}
,\nu_{str}^{(2,s)}\geq 0}(\nu_{str}^{(2,s)}(2\beta_{str}-1)+ \gamma_{str}^{(s)}+\frac{1}{c_3^{(s)}}\log(I)).\label{eq:Istrthmstr}
\end{equation}
If $\alpha$ and $\betastr$ are such that
\begin{equation}
\min_{c_3^{(s)}\geq 0}\left (-\frac{c_3^{(s)}}{2}+I_{str}(c_3^{(s)},\beta_{str})+I_{sph}(c_3^{(s)},\alpha)\right )<0,\label{eq:strcondthmstr}
\end{equation}
then the solution of (\ref{eq:l1}) is with overwhelming
probability the $k$-sparse $\x$ in (\ref{eq:system}).\label{thm:thmstrthrlift}
\end{theorem}
\begin{proof}
Follows from the above discussion.
\end{proof}

The results for the strong threshold obtained from the above theorem as well as the corresponding ones from \cite{DonohoSigned,DT,StojnicCSetam09}
are presented in Figure \ref{fig:str}.
\begin{figure}[htb]
\centering
\centerline{\epsfig{figure=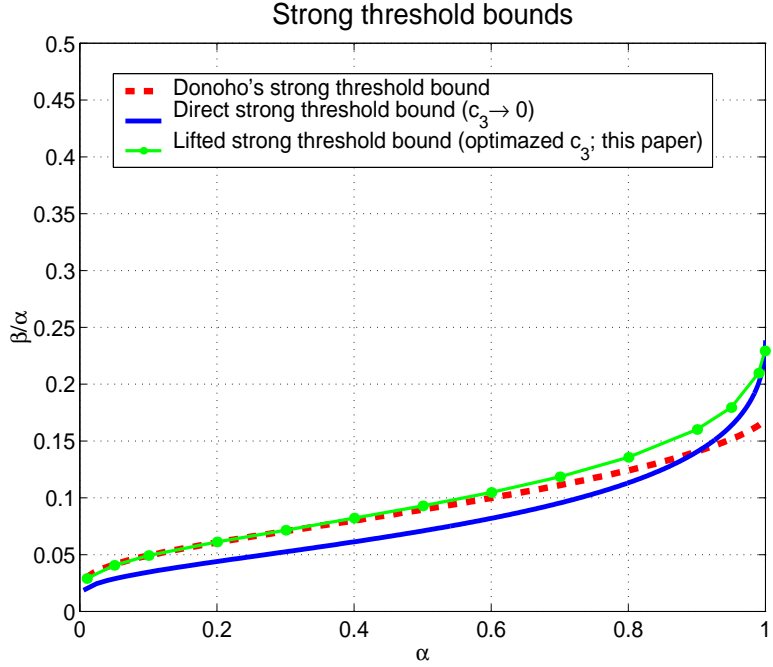,width=10.5cm,height=9cm}}
\caption{\emph{Strong} threshold, $\ell_1$-optimization}
\label{fig:str}
\end{figure}
The results substantially improve on those presented in \cite{StojnicCSetam09}. In fact they improve even on those from \cite{DonohoSigned,DT}. However, since the results are very close to the ones given in \cite{DonohoSigned,DT} we present in Tables \ref{tab:strtab1} and \ref{tab:strtab2} the concrete values of the attainable strong thresholds obtained through the mechanisms presented in \cite{DonohoSigned,DT} as well as those obtained through the mechanism presented in the above theorem. We also mention that the results presented in \cite{StojnicCSetam09} (and given in Theorem \ref{thm:thmstrthr}) can in fact be deduced from the above theorem. Namely, in the limit $c_3^{(s)}\rightarrow 0$, one from (\ref{eq:chneg9}) and (\ref{eq:defIs}) has $\max_{\w\in\Sstr}\h^T\w<\sqrt{\alpha n}$ as the limiting condition which is exactly the same condition considered in \cite{StojnicCSetam09}. For the completeness we present those results in the following corollary. (Obtaining them from the above theorem requires some work but is not that hard; instead one can just simply handle directly $\max_{\w\in\Sstr}\h^T\w$.)

\begin{corollary}(Strong threshold - lower bound)
Assume the setup of Theorem \ref{thm:thmstrthrlift}. Let $c_3^{(s)}\rightarrow 0$. Then
if $\alpha$ and $\betastr$ are such that
\begin{eqnarray}
\min_{\nu_{str}^{(1)}\geq 0}
\sqrt{\left ((1+(\nu_{str}^{(1)})^2)\mbox{erfc}(\frac{\nu_{str}^{(1)}}{\sqrt{2}})
+\frac{4\nu_{str}^{(1)}}{\sqrt{2\pi}}(2e^{-(\mbox{erfinv}(1-\betastr))^2}-\frac{1}{2}e^{-\frac{\nu_{str}^{(1)}}{2}})\right )}<\sqrt{\alpha},\nonumber\\
\label{eq:strcondcorstr}
\end{eqnarray}
then the solution of (\ref{eq:l1}) is with overwhelming
probability the $k$-sparse $\x$ in (\ref{eq:system}).\label{thm:corstrthrlift}
\end{corollary}
\begin{proof}
Theorem \ref{thm:thmstrthrlift} holds for any $c_3^{(s)}\geq 0$. The above corollary instead of looking for the best possible $c_3^{(s)}$ in Theorem \ref{thm:thmstrthrlift} assumes a simple $c_3^{(s)}\rightarrow 0$ scenario.

Alternatively, one can look at $\frac{E\max_{\w\in\Sstr}\h^T\w}{\sqrt{n}}$ and following the methodology presented in (\ref{eq:seceq1}), (\ref{eq:streq01}), and  (\ref{eq:streq02}) (and originally in \cite{StojnicCSetam09}) obtain for two scalars $c_{\nu}\geq \nu_{str}^{(1)}\geq 0$
\begin{equation}
\frac{E\max_{\w\in\Sstr}\h^T\w}{\sqrt{n}}\leq \sqrt{E_{|\h_i|\geq c_{\nu}}(|\h_i|+\nu_{str}^{(1)})^2+ E_{\nu_{str}^{(1)}\leq |\h_i|< c_{\nu}}(|\h_i|-\nu_{str}^{(1)})^2},\label{eq:corrstr0}
\end{equation}
where $c_{\nu}$ is obtained from $\beta_{str}=\int_{|\h_i|\geq c_{\nu}}\frac{e^{-\frac{\h_i^2}{2}}d\h_i}{\sqrt{2\pi}}$. Clearly, $c_{\nu}=\sqrt{2}\mbox{erfinv}(1-\beta_{str})$.
Optimizing (tightening) over $c_{\nu}\geq \nu_{str}^{(1)}\geq 0$ (and using all the concentrating machinery of \cite{StojnicCSetam09}) one can write
\begin{equation}
\frac{E\max_{\w\in\Sstr}\h^T\w}{\sqrt{n}}\doteq\min_{c_{\nu}\geq \nu_{str}^{(1)}\geq 0}\sqrt{\left (\int_{|\h_i|\geq c_{\nu}}(|\h_i|+\nu_{str}^{(1)})^2\frac{e^{-\frac{\h_i^2}{2}}d\h_i}{\sqrt{2\pi}}
+\int_{\nu_{str}^{(1)}\leq |\h_i|< c_{\nu}}(|\h_i|-\nu_{str}^{(1)})^2\frac{e^{-\frac{\h_i^2}{2}}d\h_i}{\sqrt{2\pi}}\right )}.\label{eq:corrstr1}
\end{equation}

On the other hand, when $c_3^{(s)}\rightarrow 0$ a combination of (\ref{eq:Isphcorsec}), (\ref{eq:streq030}), (\ref{eq:gamaiden2str}), and (\ref{eq:strcondthmstr}) gives
\begin{equation}
\min_{\gamma_{str}^{(s)},\nu_{str}^{(1)}
,\nu_{str}^{(2,s)}\geq 0}(\nu_{str}^{(2,s)}(2\beta_{str}-1)+ \gamma_{str}^{(s)}+E\t_i^{(s)})\leq \sqrt{\alpha},\label{eq:gamaiden2strcorpr}
\end{equation}
where
\begin{equation}
\t_i^{(s)}=\max\left (\frac{(|\h_i|+\nu_{str}^{(1)})^2}{4\gamma_{str}^{(s)}}-\nu_{str}^{(2,s)},
\frac{(\max(|\h_i|-\nu_{str}^{(1)},0))^2}{4\gamma_{str}^{(s)}}+\nu_{str}^{(2,s)}\right ).\label{eq:streq030corpr}
\end{equation}
A particular choice $\nu_{str}^{(2,s)}=\frac{2c_{\nu}\nu_{str}^{(1)}}{4\gamma_{str}^{(s)}}$ (where as above, $c_{\nu}=\sqrt{2}\mbox{erfinv}(1-\beta_{str})$) transforms (\ref{eq:gamaiden2strcorpr}) to
\begin{equation}
\min_{\gamma_{str}^{(s)},\nu_{str}^{(1)}\geq 0}
\left (\gamma_{str}^{(s)}+\frac{\left (\int_{|\h_i|\geq c_{\nu}}(|\h_i|+\nu_{str}^{(1)})^2\frac{e^{-\frac{\h_i^2}{2}}d\h_i}{\sqrt{2\pi}}
+\int_{\nu_{str}^{(1)}\leq |\h_i|< c_{\nu}}(|\h_i|-\nu_{str}^{(1)})^2\frac{e^{-\frac{\h_i^2}{2}}d\h_i}{\sqrt{2\pi}}\right )}{4\gamma_{str}^{(s)}}\right )\leq \sqrt{\alpha},\label{eq:gamaiden2strcorpr1}
\end{equation}
which after optimization over $\gamma_{str}^{(s)}$ becomes
\begin{equation}
\min_{\nu_{str}^{(1)}\geq 0}
\sqrt{\left (\int_{|\h_i|\geq c_{\nu}}(|\h_i|+\nu_{str}^{(1)})^2\frac{e^{-\frac{\h_i^2}{2}}d\h_i}{\sqrt{2\pi}}
+\int_{\nu_{str}^{(1)}\leq |\h_i|< c_{\nu}}(|\h_i|-\nu_{str}^{(1)})^2\frac{e^{-\frac{\h_i^2}{2}}d\h_i}{\sqrt{2\pi}}\right )}\leq \sqrt{\alpha}.\label{eq:gamaiden2strcorpr2}
\end{equation}
This is exactly what one gets from (\ref{eq:corrstr1}). Solving the integrals from (\ref{eq:corrstr1}) or (\ref{eq:gamaiden2strcorpr2}) then leads to the condition given in the above corollary. At the same time this also shows that in (\ref{eq:streq04}) one, for all statistical purposes, indeed has an equality (of course much more is true but the proof of that would just make our presentation more cumbersome and those statements are not necessary in a statistical context that we consider here).
\end{proof}

\noindent \textbf{Remark:} Solving (\ref{eq:strcondcorstr}) over $\nu_{str}^{(1)}$ and juggling a bit one can arrive at the results presented in Theorem \ref{thm:thmstrthr}. Of course, the results of Theorem \ref{thm:thmstrthr} were presented in more detail in \cite{StojnicCSetam09} where our main concern was a thorough studying of the underlying optimization problem.

The results obtained from the previous corollary (and obviously the results from Theorem \ref{thm:thmstrthr}) are those presented in Figure \ref{fig:str} and Tables \ref{tab:strtab1} and \ref{tab:strtab2} that we refer to as $c_3\rightarrow 0$ scenario or direct strong threshold bounds (also we remove superscript $(s)$ from $c_3^{(s)}$ to make the figure and tables easier to view). As can be see from Figure \ref{fig:str} and Tables \ref{tab:strtab1} and \ref{tab:strtab2}, a substantial improvement over the strong threshold results from \cite{StojnicCSetam09} can be achieved through the above introduced mechanism.

While any improvement in strong thresholds is in our opinion quite a feat, one can also observe that ultimately the results we presented above
don't make a big improvement over those obtained in \cite{DonohoUnsigned,DT}. That can be because our methods are not strong enough for such an improvement or simply because a bigger improvement may not be possible (in other words the results obtained in \cite{DonohoUnsigned,DT} may very well already be fairly close to the optimal ones). As for the limits of the developed methods, we do want to emphasize that (as in the case of the sectional thresholds) we did solve the numerical optimizations that appear in Theorem \ref{thm:thmstrthrlift} only on a local optimum level and obviously only with a finite precision. We do not know if a substantial change in the presented results would occur had we solved it on a global optimum level (we again recall that finding local optima is, of course, certainly enough to establish attainable values of the strong thresholds). As for how far away from the true strong thresholds are the results presented in the tables and Figure \ref{fig:str}, we actually believe that they are pretty close.

\begin{table}
\caption{Strong threshold bounds -- low $\alpha\leq 0.5$ regime}\vspace{.1in}
\hspace{-0in}\centering
\begin{tabular}{||c|c|c|c|c|c|c|c||}\hline\hline
 $\alpha$  & $0.01$ & $0.05$ & $0.1$ & $0.2$ & $0.3$ & $0.4$ & $0.5$ \\ \hline\hline
 $\beta_{str}$ (Donoho \cite{DonohoPol,DonohoUnsigned})  & $ 0.00031 $ & $  0.00205  $ & $ 0.00488$ & $
0.01250$ & $ 0.02109  $ & $ 0.03192
$ & $0.04471$  \\ \hline
 $\beta_{str}$ (Theorem \ref{thm:thmstrthrlift})
 & $0.00030$ & $0.00206 $ & $0.00492$ & $0.01225$ & $0.02154 $ & $0.03285 $ & $0.04645$\\ \hline\hline
\end{tabular}
\label{tab:strtab1}
\end{table}

\begin{table}
\caption{Strong threshold bounds -- high $\alpha> 0.5$ regime}\vspace{.1in}
\hspace{-0in}\centering
\begin{tabular}{||c|c|c|c|c|c|c|c|c||}\hline\hline
 $\alpha$  & $0.6$ & $0.7$ & $0.8$ & $0.9$ & $0.95$ & $0.99$ & $0.999$ & $0.9999$  \\ \hline\hline
 $\beta_{str}$ (Donoho \cite{DonohoPol,DonohoUnsigned}) & $ 0.05977 $ & $  0.07760$ & $
0.1000 $ & $  0.1264 $ & $0.1438$ & $
0.1620$ & $ 0.1677$  & $0.1685$ \\ \hline
 $\beta_{str}$ (Theorem \ref{thm:thmstrthrlift})
  & $0.06287$ & $
0.08298$ & $0.1085$ & $0.1443$ & $0.1710$ & $0.2080$ & $0.2291$ & $0.2359$\\ \hline\hline
\end{tabular}
\label{tab:strtab2}
\end{table}

\section{Non-negative $\x$}
\label{sec:backnon}

In this section we will look at a subcase of the general under-determined linear system with sparse solutions. Namely, we will restrict our attention to the case when it is a priori known that the unknown $\x$ in (\ref{eq:system}) has only non-negative components. One is then interested in finding a $k$-sparse $\x$ such
that
\begin{equation}
A\x=\y, \x_i\geq 0,1\leq i\leq n. \label{eq:systemnon}
\end{equation}
where as earlier, $A$ is an $m\times n$ ($m<n$) matrix and $\y$ is
an $m\times 1$ vector (here we again under $k$-sparse vector we assume a vector that has at most $k$ nonzero
components). Of course, the assumption will be that such an $\x$ exists.

One can then employ a simple adaptation of the basic $\ell_1$-optimization algorithm  from (\ref{eq:l1}) to find $\x$ in
(\ref{eq:systemnon})
\begin{eqnarray}
\mbox{min} & & \|\x\|_{1}\nonumber \\
\mbox{subject to} & & A\x=\y\nonumber \\
& & \x_i\geq 0,1\leq i\leq n.\label{eq:l1non}
\end{eqnarray}
Of course, this case has already been thoroughly studied in the literature primarily through the work of Donoho and Tanner \cite{DonohoSigned,DT}. Namely,
in \cite{DonohoSigned,DT} Donoho  and Tanner consider the polytope obtained by
projecting the regular $n$-dimensional simplex $T_p^n$ by $A$. They then established that
the solution of (\ref{eq:l1non}) will be the $k$-sparse nonnegative solution of
(\ref{eq:systemnon}) (or (\ref{eq:system})) if and only if
$AT_p^n$ is $k$-neighborly
(for the definitions of neighborliness, details of Donoho and Tanner's approach, and related results the interested reader can consult now already classic references \cite{DonohoUnsigned,DonohoPol,DonohoSigned,DT}). In a nutshell, using the results
of \cite{PMM,AS,BorockyHenk,Ruben,VS}, it is shown in
\cite{DT,DonohoSigned}, that if $A$ is a random $m\times n$
ortho-projector matrix then with overwhelming probability $AT_p^n$ is $k$-neighborly (as earlier, under overwhelming probability we in this paper assume
a probability that is no more than a number exponentially decaying in $n$ away from $1$). Miraculously, \cite{DT,DonohoSigned} provided a precise characterization of $m$ and $k$ (in a large dimensional context) for which this happens.

As was the case when we discussed the general vectors $\x$, it should be noted that one usually considers success of
(\ref{eq:l1non}) in recovering \emph{any} given a priori known to have only non-negative components $k$-sparse $\x$ in (\ref{eq:systemnon}). It is also of interest to consider success of
(\ref{eq:l1non}) in recovering
\emph{almost any} such given $\x$ in (\ref{eq:systemnon}). We below make a distinction between these
cases and recall on some of the definitions from
\cite{DonohoPol,DT,DTciss,DTjams2010,StojnicCSetam09,StojnicICASSP09}.

Clearly, for any given constant $\alpha\leq 1$ there is a maximum
allowable value of $\beta$ such that for \emph{any} given a priori known to have only non-negative components $k$-sparse $\x$ in (\ref{eq:system}) the solution of (\ref{eq:l1})
is with overwhelming probability exactly that given $k$-sparse $\x$. We will refer to this maximum allowable value of
$\beta$ as the \emph{strong threshold} (see
\cite{DT,DonohoSigned}) and will denote it as $\beta_{str}^{+}$. Similarly, for any given constant
$\alpha\leq 1$ and \emph{any} given a priori known to have only non-negative components $\x$ with a given fixed location of non-zero components
there will be a maximum allowable value of $\beta$ such that
(\ref{eq:l1non}) finds that given $\x$ in (\ref{eq:systemnon}) with overwhelming
probability. We will refer to this maximum allowable value of
$\beta$ as the \emph{weak threshold} and will denote it by $\beta_{w}^{+}$ (see, e.g. \cite{StojnicICASSP09,StojnicCSetam09}).

When viewed within this frame the results of \cite{DonohoPol,DonohoUnsigned} established the exact values of $\beta_w^{+}$ and provided lower bounds on $\beta_{str}^{+}$.

In a series of our own work (see, e.g. \cite{StojnicICASSP09,StojnicCSetam09,StojnicUpper10}) we then created an alternative probabilistic approach which was capable of providing the precise characterization of $\beta_w^{+}$ as well and thereby of reestablishing the results of Donoho and Tanner \cite{DonohoSigned,DT} through a purely probabilistic approach. We also mentioned in \cite{StojnicCSetam09} that it would not be that hard to establish analogous lower bounds on $\beta_{str}^{+}$ (however, since those were not as successful as the counterparts that we obtained for general vectors $\x$ we chose not to state them explicitly in \cite{StojnicCSetam09}). Below, we will present both of these results. As was the case when we studied general $\x$ in earlier sections, fairly often we will use these theorems as a benchmark for the results that we will present in this paper.

The first of the theorems relates to the weak threshold $\beta_w^{+}$.

\begin{theorem}(Weak threshold ($\x$ a priori known to be non-negative); -- exact \cite{StojnicCSetam09,StojnicUpper10})
Let $A$ be an $m\times n$ matrix in (\ref{eq:system})
with i.i.d. standard normal components. Let
the unknown $\x$ in (\ref{eq:systemnon}) be $k$-sparse and let it be a priori known that its nonzero components are positive. Further, let the location of the nonzero elements of $\x$ be arbitrarily chosen but fixed.
Let $k,m,n$ be large
and let $\alpha=\frac{m}{n}$ and $\beta_w^+=\frac{k}{n}$ be constants
independent of $m$ and $n$. Let $\erfinv$ be the inverse of the standard error function associated with zero-mean unit variance Gaussian random variable.  Further,
let all $\epsilon$'s below be arbitrarily small constants.
\begin{enumerate}
\item Let $\htheta_w^+$, ($\beta_w^+\leq \htheta_w^+\leq 1$) be the solution of
\begin{equation}
(1-\epsilon_1^{(c)})(1-\betawnon)\frac{\sqrt{\frac{1}{2\pi}}e^{-(\erfinv(2\frac{1-\thetawnon}{1-\betawnon}-1))^2}}{\thetawnon}-\sqrt{2}\erfinv ((2\frac{(1+\epsilon_1^{(c)})(1-\thetawnon)}{1-\betawnon}-1))=0.\label{eq:thmweaknontheta1}
\end{equation}
If $\alpha$ and $\beta_w^+$ further satisfy
\begin{equation}
\hspace{-.7in}\alpha>\frac{1-\betawnon}{\sqrt{2\pi}}\left (\frac{\sqrt{2(\erfinv(2\frac{1-\hthetawnon}{1-\betawnon}-1))^2}}{e^{(\erfinv(2\frac{1-\hthetawnon}{1-\betawnon}-1))^2}}\right )+\hthetawnon
-\frac{\left ((1-\betawnon)\sqrt{\frac{1}{2\pi}}e^{-(\erfinv(2\frac{1-\hthetawnon}{1-\betawnon}-1))^2}\right )^2}{\hthetawnon}\label{eq:thmweakalphanon}
\end{equation}
then with overwhelming probability the solution of (\ref{eq:l1non}) is the nonnegative $k$-sparse $\x$ from (\ref{eq:systemnon}).
\item Let $\htheta_w^+$, ($\beta_w^+\leq \htheta_w^+\leq 1$) be the solution of
\begin{equation}
(1+\epsilon_2^{(c)})(1-\betawnon)\frac{\sqrt{\frac{1}{2\pi}}e^{-(\erfinv(2\frac{1-\thetawnon}{1-\betawnon}-1))^2}}{\thetawnon}-\sqrt{2}\erfinv ((2\frac{(1-\epsilon_2^{(c)})(1-\thetawnon)}{1-\betawnon}-1))=0.\label{eq:thmweaknontheta2}
\end{equation}
If on the other hand $\alpha$ and $\beta_w^+$ satisfy
\begin{multline}
\hspace{-1in}\alpha<\frac{1}{(1+\epsilon_{1}^{(m)})^2}\left ((1-\epsilon_{1}^{(g)})(\htheta_w^++\frac{(1-\beta_w^+)}{\sqrt{2\pi}} \frac{\sqrt{2(\erfinv(2\frac{1-\htheta_w^+}{1-\beta_w^+}-1))^2}}{e^{(\erfinv(2\frac{1-\htheta_w^+}{1-\beta_w^+}-1))^2}})
-\frac{\left ((1-\beta_w^+)\sqrt{\frac{1}{2\pi}}e^{-(\erfinv(2\frac{1-\hat{\theta}_w^+}{1-\beta_w^+}-1))^2}\right )^2}{\hat{\theta}_w^+(1+\epsilon_{3}^{(g)})^{-2}}\right )\label{eq:thmweaknonalpha2}
\end{multline}
then with overwhelming probability there will be a nonnegative $k$-sparse $\x$ (from a set of $\x$'s with fixed locations of nonzero components) that satisfies (\ref{eq:systemnon}) and is \textbf{not} the solution of (\ref{eq:l1non}).
\end{enumerate}
\label{thm:thmweaknonthr}
\end{theorem}
\begin{proof}
The first part was established in \cite{StojnicCSetam09}. The second part was established in \cite{StojnicUpper10}. An alternative proof was also presented in \cite{StojnicEquiv10}.
\end{proof}

\noindent As was the case with the weak thresholds for general vectors $\x$, the previous theorem insists on precision and involves ``epsilon" type of characterization. However, one can again do what we, in Section \ref{sec:back}, referred to as the ``deepsilonification" and obtain a way more convenient characterization. After removing all $\epsilon$'s one in a more informal language then has.

\noindent Assume the setup of the above theorem. Let $\alpha_w^+$ and $\beta_w^+$ satisfy the following:

\noindent \underline{\underline{\textbf{Fundamental characterization of the $\ell_1$ performance ($\x$ in (\ref{eq:systemnon}) a priori known to be nonnegative):}}}

\begin{center}
\shadowbox{$
(1-\beta_w^+)\frac{\sqrt{\frac{1}{2\pi}}e^{-(\erfinv(2\frac{1-\alpha_w^+}{1-\beta_w^+}-1))^2}}{\alpha_w^+}-\sqrt{2}\erfinv (2\frac{1-\alpha_w^+}{1-\beta_w^+}-1)=0.
$}
-\vspace{-.5in}\begin{equation}
\label{eq:thmweaknontheta2}
\end{equation}
\end{center}

Then:
\begin{enumerate}
\item If $\alpha>\alpha_w^+$ then with overwhelming probability the solution of (\ref{eq:l1non}) is the a priori known to be nonnegative $k$-sparse $\x$ from (\ref{eq:systemnon}).
\item If $\alpha<\alpha_w^+$ then with overwhelming probability there will be an a priori known to be nonnegative $k$-sparse $\x$ (from a set of $\x$'s with fixed locations of nonzero components) that satisfies (\ref{eq:systemnon}) and is \textbf{not} the solution of (\ref{eq:l1non}).
    \end{enumerate}

The following theorem summarizes the results related to the strong thresholds ($\beta_{str}^{+}$).

\begin{theorem}(Strong threshold ($\x$ a priori known to be non-negative) -- lower bound)
Let $A$ be an $m\times n$ measurement matrix in (\ref{eq:system})
with i.i.d. standard normal components. Let
the unknown $\x$ in (\ref{eq:systemnon}) be $k$-sparse and let it be a priori known that its nonzero components are positive. Let $k,m,n$ be large
and let $\alpha^+=\frac{m}{n}$ and $\beta_{str}^{+}=\frac{k}{n}$ be constants
independent of $m$ and $n$. Let $\erfinv$ be the inverse of the standard error function associated with zero-mean unit variance Gaussian random variable. Further, let $\epsilon^+$ be an arbitrarily small positive constant and let $\htheta_s^{+}$ ($0\leq \htheta_s^{+}\leq 1-\beta_{str}^{+}$) be the solution of
\begin{equation}
(1-\epsilon^+)\sqrt{\frac{1}{2\pi}}\frac{e^{-(\mbox{erfinv}(2(1-\htheta_{s}^{+})-1))^2}
-e^{-(\mbox{erfinv}(2(1-\beta_{str}^{+})-1))^2}}{\htheta_s^{+}+\beta_{str}^{+}}-\sqrt{2}\mbox{erfinv}(2((1+\epsilon^+)(1-\htheta_s^{+}))-1)=0.\label{eq:thmstrthetanon}
\end{equation}
Also, set
\begin{eqnarray}
S_{1}^{+} & = & 1/2/\sqrt{2\pi}\left (2\sqrt{2\pi}\htheta_s^{+}+\frac{2\sqrt{2}(\mbox{erfinv}(2(1-\htheta_{s}^{+})-1))}
{e^{(\mbox{erfinv}(2(1-\htheta_{s}^{+})-1))^2}}\right )\nonumber \\
S_{2}^{+} & = & 1/2/\sqrt{2\pi}\left (2\sqrt{2\pi}\beta_{str}^{+}+\frac{2\sqrt{2}(\mbox{erfinv}(2(1-\beta_{str}^{+})-1))}
{e^{(\mbox{erfinv}(2(1-\beta_{str}^{+})-1))^2}}\right )\end{eqnarray}
If $\alpha^+$ and $\beta_{str}^{+}$ further satisfy
\begin{equation}
\alpha^{+}>\left (S_1^{+}+S_2^{+}-\frac{(\sqrt{\frac{1}{2\pi}}e^{-(\mbox{erfinv}(2(1-\htheta_{s}^{+})-1))^2}
-\sqrt{\frac{1}{2\pi}}e^{-(\mbox{erfinv}(2(1-\beta_{str}^{+})-1))^2})^2}{\htheta_s^{+}+\beta_{str}^{+}}\right )\label{eq:thmstralphanon}
\end{equation}
then with overwhelming probability the solution of (\ref{eq:l1non}) is the $k$-sparse $\x$ from (\ref{eq:systemnon}).\label{thm:thmstrthrnon}
\end{theorem}
\begin{proof}
The proof is essentially following the methodology developed in \cite{StojnicCSetam09}. However, it is a bit tedious and since it is not the main concern of this paper we only state the results and skip the details of the proof.
\end{proof}

We will show the results for the strong thresholds one can obtain through the above theorem in subsequent sections when we discuss the corresponding ones obtained in this paper. As for what will be the main topic below, essentially we will attempt to translate the mechanism from the previous section so that it fits the a priori known to be nonnegative vectors $\x$. It will turn out that the mechanism we develop provides a substantial improvement over the results from Theorem \ref{thm:thmstrthrnon}. Moreover, the results it provides will improve even on those of \cite{DonohoSigned,DT}.

\section{Lifting $\ell_1$-minimization strong threshold; non-negative $\x$}
\label{sec:strthrnon}

In this section we look at the strong thresholds when $\x$ in (\ref{eq:systemnon}) is a priori known to be nonnegative. Of course, there is really nothing specific about this choice of signs. One can choose any set of signs and all our results will hold; the difference though is that it may not be easy to run the algorithms with a collection of different signs (still if one could do it, all our results would hold for such algorithms as well). We do recall as in Section \ref{sec:strthr} that these strong threshold results were substantially harder to establish than the ones that we presented in Section \ref{sec:secthr} (and one may say a bit harder even than those established in Section \ref{sec:strthr}).

As in the previous sections, throughout the presentation in this section we will assume a substantial level of familiarity with many of the well-known results that relate to the performance characterization of (\ref{eq:l1non}) (we will again fairly often recall on many results/definitions that we presented in \cite{StojnicCSetam09}). We start by defining a set $\Sstrnon$
\begin{equation}
\hspace{-.3in}\Sstrnon=\{\w\in S^{n-1}| \quad \w_i=\b_i\w_i^{\prime},\sum_{i=1}^n \b_i\w_i^{\prime}<0,\b_i^2=1,\sum_{i=1}^n\b_i=n-2k;\w_i^{\prime}\geq 0 \quad \mbox{if}\quad \b_i=1\},\label{eq:defSstrnon}
\end{equation}
where $S^{n-1}$ is the unit sphere in $R^n$. Then according to what was established in \cite{StojnicCSetam09} the following optimization problem is of critical importance in determining the strong threshold of $\ell_1$-minimization
\begin{equation}
\xi_{str}^{+}=\min_{\w\in\Sstrnon}\|A\w\|_2.\label{eq:negham1strnon}
\end{equation}
We recall that what was established in \cite{StojnicCSetam09} is roughly the following: if $\xi_{str}^{+}$ is positive with overwhelming probability for certain combination of $k$, $m$, and $n$ then for $\alpha=\frac{m}{n}$ one has a lower bound $\beta_{str}^{+}=\frac{k}{n}$ on the true value of the strong threshold with overwhelming probability. Also, as was the case with thresholds studied in the previous section, the mechanisms of \cite{StojnicCSetam09} were powerful enough to establish the concentration of $\xi_{str}^{+}$ as well. This essentially means that if we can show that $E\xi_{str}^{+}>0$ for certain $k$, $m$, and $n$ we can then obtain the lower bound on the strong threshold. This is, of course, precisely what was done in Theorem \ref{thm:thmstrthrnon}. However, as was the case with the sectional and strong thresholds of general vectors $\x$, the results we obtained for the strong threshold in Theorem \ref{thm:thmstrthrnon} through such a consideration are not exact (in fact, as we will see below, they are quite below known lower bounds). The main reason of course was inability to determine $E\xi_{str}^{+}$ exactly. Instead we resorted to its lower bounds and similarly to what happened when we consider sectional and strong thresholds of general vectors, those bounds turned out to be loose. In this section we will use some of the ideas from the previous sections (which as we stated earlier have roots in the mechanisms we recently introduced in \cite{StojnicMoreSophHopBnds10}) to provide a substantial conceptual improvement in these bounds. As in earlier sections, this will in turn reflect in a conceptual and substantial practical improvement of the strong thresholds. We again emphasize that although the translation that we will present below will appear seemingly smooth it was not so trivial to achieve it.

Below we present a way to create a lower-bound on the optimal value of (\ref{eq:negham1strnon}).

\subsection{Lower-bounding $\xi_{str}^{+}$}
\label{sec:lbxistrnon}

As mentioned above, in this subsection we will look at problem from (\ref{eq:negham1strnon}) or more precisely its optimal value.

We start by reformulating the problem in (\ref{eq:negham1strnon}) in the following way
\begin{equation}
\xi_{str}^{+}=\min_{\w\in\Sstrnon}\max_{\|\y\|_2=1}\y^TA\w.\label{eq:sqrtnegham2strnon}
\end{equation}
Then we continue by reformulating Lemma \ref{lemma:negexplemmastr}. Namely, based on Lemma \ref{lemma:negexplemmastr}, we establish the following:
\begin{lemma}
Let $A$ be an $m\times n$ matrix with i.i.d. standard normal components. Let $\g$ and $\h$ be $n\times 1$ and $m\times 1$ vectors, respectively, with i.i.d. standard normal components. Also, let $g$ be a standard normal random variable and let $c_3$ be a positive constant. Then
\begin{equation}
E(\max_{\w\in\Sstrnon}\min_{\|\y\|_2=1}e^{-c_3(\y^T A\w + g)})\leq E(\max_{\w\in\Sstrnon}\min_{\|\y\|_2=1}e^{-c_3(\g^T\y+\h^T\w)}).\label{eq:negexplemmanon}
\end{equation}\label{lemma:negexplemmastrnon}
\end{lemma}
\begin{proof}
The proof is a standard/direct application of Theorem \ref{thm:Gordonneg1} (as was the proof of Lemmas \ref{lemma:negexplemma} and \ref{lemma:negexplemmastr}). As earlier, we will of course omit the details and just mention that the only difference between this lemma and Lemmas \ref{lemma:negexplemma} and \ref{lemma:negexplemmastr} is in the structure of set $\Sstrnon$. What is here $\Sstrnon$ it is $\Ssec$ in Lemma \ref{lemma:negexplemma} and
$\Sstr$ in Lemma \ref{lemma:negexplemmastr}. However, such a difference introduces no structural changes in the proof.
\end{proof}

Following what was done in the previous section we arrive at the following analogue of (\ref{eq:chneg8}) and (\ref{eq:chneg8str}) (and ultimately of \cite{StojnicMoreSophHopBnds10}'s equation $(57)$):
\begin{equation}
E(\min_{\w\in\Sstr}\|A\w\|_2)\geq
\frac{c_3}{2}-\frac{1}{c_3}\log(E(\max_{\w\in\Sstr}(e^{-c_3\h^T\w})))
-\frac{1}{c_3}\log(E(\min_{\|\y\|_2=1}(e^{-c_3\g^T\y}))).\label{eq:chneg8strnon}
\end{equation}
As earlier, let $c_3=c_3^{(s)}\sqrt{n}$ where $c_3^{(s)}$ is a constant independent of $n$. Then (\ref{eq:chneg8str}) becomes
\begin{eqnarray}
\hspace{-.5in}\frac{E(\min_{\w\in\Sstr}\|A\w\|_2)}{\sqrt{n}}
& \geq &
\frac{c_3^{(s)}}{2}-\frac{1}{nc_3^{(s)}}\log(E(\max_{\w\in\Sstr}(e^{-c_3^{(s)}\sqrt{n}\h^T\w})))
-\frac{1}{nc_3^{(s)}}\log(E(\min_{\|\y\|_2=1}(e^{-c_3^{(s)}\sqrt{n}\g^T\y})))\nonumber \\
& = &-(-\frac{c_3^{(s)}}{2}+I_{str}^{+}(c_3^{(s)},\beta)+I_{sph}(c_3^{(s)},\alpha)),\label{eq:chneg9strnon}
\end{eqnarray}
where
\begin{eqnarray}
I_{str}^{+}(c_3^{(s)},\beta) & = & \frac{1}{nc_3^{(s)}}\log(E(\max_{\w\in\Sstr}(e^{-c_3^{(s)}\sqrt{n}\h^T\w})))\nonumber \\
I_{sph}(c_3^{(s)},\alpha) & = & \frac{1}{nc_3^{(s)}}\log(E(\min_{\|\y\|_2=1}(e^{-c_3^{(s)}\sqrt{n}\g^T\y}))).\label{eq:defIsstrnon}
\end{eqnarray}

As in the previous sections, one should now note that the above bound is effectively correct for any positive constant $c_3^{(s)}$. The only thing that is then left to be done so that the above bound becomes operational is to estimate $I_{str}^{+}(c_3^{(s)},\beta)$ and $I_{sph}(c_3^{(s)},\alpha)$. We again recall that
\begin{equation}
I_{sph}(c_3^{(s)},\alpha)=\frac{1}{nc_3^{(s)}}\log(Ee^{-c_3^{(s)}\sqrt{n}\|\g\|_2})\doteq
\left ( \widehat{\gamma_{sph}^{(s)}}-\frac{\alpha}{2c_3^{(s)}}\log(1-\frac{c_3^{(s)}}{2\widehat{\gamma_{sph}^{(s)}}}\right ),\label{eq:Isphstrnon}
\end{equation}
where
\begin{equation}
\widehat{\gamma_{sph}^{(s)}}=\frac{2c_3^{(s)}-\sqrt{4(c_3^{(s)})^2+16\alpha}}{8}.\label{eq:gamaiden3strnon}
\end{equation}

We now switch to $I_{str}^{+}(c_3^{(s)},\beta)$. Similarly to what was stated earlier, pretty good estimates for this quantity can be obtained for any $n$. However, to facilitate the exposition we will focus only on the large $n$ scenario. In that case one can again use the saddle point concept applied in \cite{SPH}. As above, we present the core idea without all the details from \cite{SPH}. Let $f(\w)=-\h^T\w$ and
we start with the following line of identities
\begin{eqnarray}
f_{str}^{+}=\max_{\w\in\Sstrnon}f(\w)=-\min_{\w\in\Sstrnon}\h^T\w =  -\min_{\b,\w} & & \sum_{i=1}^{n}\h_i\b_i\w_i^{\prime}\nonumber \\
\mbox{subject to} & & \sum_{i=1}^{n} \b_i\w_i^{\prime}<0,\nonumber \\
& & \sum_{i=1}^{n} (\w_i^{\prime})^2=1,\nonumber \\
& & \b_i\in\{-1,1\},1 \leq i\leq n,\nonumber \\
& & \sum_{i=1}^{n}\b_i= n-2k, \nonumber \\
& & \b_i=1\Rightarrow \w_i^{\prime}\geq 0,1\leq i\leq n.\nonumber \\\label{eq:streq0non}
\end{eqnarray}
We then have
\begin{multline}
\hspace{-.5in}f_{str}^{+}  =  -\min_{\b_i^2=1,\w^{\prime},\b_i=1\Rightarrow \w_i^{\prime}\geq 0}\max_{\gamma_{str},\nu_{str}^{(1)},\nu_{str}^{(2)}\geq 0} \sum_{i=1}^{n}\h_i\b_i\w_i^{\prime}+\nu_{str}^{(1)}\sum_{i=1}^{n} \b_i\w_i^{\prime}
-\nu_{str}^{(2)}\sum_{i=1}^{n}\b_i+\nu_{str}^{(2)}(n-2k)+\gamma_{str}\sum_{i=1}^{n} (\w_i^{\prime})^2-\gamma_{str} \\
\leq  -\max_{\gamma_{str},\nu_{str}^{(1)},\nu_{str}^{(2)}\geq 0}\min_{\b_i^2=1,\w^{\prime},\b_i=1\Rightarrow \w_i^{\prime}\geq 0} \sum_{i=1}^{n}\h_i\b_i\w_i^{\prime}+\nu_{str}^{(1)}\sum_{i=1}^{n} \b_i\w_i^{\prime}
-\nu_{str}^{(2)}\sum_{i=1}^{n}\b_i+\nu_{str}^{(2)}(n-2k)+\gamma_{str}\sum_{i=1}^{n} (\w_i^{\prime})^2-\gamma_{str} \\
=  \min_{\gamma_{str},\nu_{str}^{(1)},\nu_{str}^{(2)}\geq 0}\max_{\b_i^2=1,\w^{\prime},\b_i=1\Rightarrow \w_i^{\prime}\geq 0} -\sum_{i=1}^{n}\h_i\b_i\w_i^{\prime}-\nu_{str}^{(1)}\sum_{i=1}^{n} \b_i\w_i^{\prime}
+\nu_{str}^{(2)}\sum_{i=1}^{n}\b_i-\nu_{str}^{(2)}(n-2k)-\gamma_{str}\sum_{i=1}^{n} (\w_i^{\prime})^2+\gamma_{str} \\
\label{eq:streq01non}
\end{multline}
Positivity condition on $\nu_{str}^{(2)}$ is added although it is not necessary (it essentially amounts to relaxing the last constraint to an inequality which changes nothing with respect to the final results).
Since we are interesting in statistical behavior change of variables $\h=-\h$ could have been done at the beginning and all of the above would hold with $-\h$. Since it will be a bit easier for us not to worry about the minus sign we will assume that this change had been done and work from this point on writing $-\h$ instead of $\h$. To solve the inner optimization it helps to introduce a vector $\t^{+}$ in the following way
\begin{equation}
\t_i^{+}=\max\left (\frac{(\h_i-\nu_{str}^{(1)})^2}{4\gamma_{str}}-\nu_{str}^{(2)},
\frac{(\max(\h_i-\nu_{str}^{(1)},0))^2}{4\gamma_{str}}+\nu_{str}^{(2)}\right ).\label{eq:streq030non}
\end{equation}
or alternatively (possibly in a more convenient way)
\begin{equation}
\t_i^{+}=\begin{cases}\frac{\h_i^2+(\nu_{str}^{(1)})^2}{4\gamma_{str}}-\frac{\h_i\nu_{str}^{(1)}}{2\gamma_{str}}-\nu_{str}^{(2)}, &  \h_i\leq \nu_{str}^{(1)}-\sqrt
{8\gamma_{str}\nu_{str}^{(2)}}\\
\nu_{str}^{(2)}, & \nu_{str}^{(1)}-\sqrt{8\gamma_{str}\nu_{str}^{(2)}} \leq \h_i\leq \nu_{str}^{(1)}\\
\frac{\h_i^2+(\nu_{str}^{(1)})^2}{4\gamma_{str}}-\frac{\h_i\nu_{str}^{(1)}}{2\gamma_{str}}+\nu_{str}^{(2)}, & \h_i\geq \nu_{str}^{(1)}.\label{eq:streq03non}
\end{cases}
\end{equation}
Using (\ref{eq:streq03non}), (\ref{eq:streq01non}) then becomes
\begin{multline}
\hspace{-.5in}f_{str}^{+}  \leq   \min_{\gamma_{str},\nu_{str}^{(1)},\nu_{str}^{(2)}\geq 0}\max_{\b_i^2=1,\w^{\prime},\b_i=1\Rightarrow \w_i^{\prime}\geq 0} \sum_{i=1}^{n}\h_i\b_i\w_i^{\prime}-\nu_{str}^{(1)}\sum_{i=1}^{n} \b_i\w_i^{\prime}
+\nu_{str}^{(2)}\sum_{i=1}^{n}\b_i-\nu_{str}^{(2)}(n-2k)-\gamma_{str}\sum_{i=1}^{n} (\w_i^{\prime})^2+\gamma_{str} \\
 =  \min_{\gamma_{str},\nu_{str}^{(1)},\nu_{str}^{(2)}\geq 0} \sum_{i=1}^{n}\t_i^{+}+\nu_{str}^{(2)}(2k-n)+\gamma_{str},
\label{eq:streq04non}
\end{multline}
where we emphasize again that $\h$ has been replaced with $-\h$ for the ease of the exposition. As when we studied general vectors $\x$, although we showed an inequality on $f_{str}^{+}$ (which is sufficient for what we need here) we do mention that the above actually holds with the equality. Let
\begin{equation}
f_1^{(str+)}(\h,\nu_{str}^{(1)},\nu_{str}^{(2)},\gamma_{str},\beta)=\sum_{i=1}^{n} \t_i^{+}.\label{eq:deff1strnon}
\end{equation}
Then
\begin{multline}
\hspace{-.3in}I_{str}^{+}(c_3^{(s)},\beta)  =  \frac{1}{nc_3^{(s)}}\log(E(\max_{\w\in\Sstrnon}(e^{-c_3^{(s)}\sqrt{n}\h^T\w}))) = \frac{1}{nc_3^{(s)}}\log(E(\max_{\w\in\Sstrnon}(e^{c_3^{(s)}\sqrt{n}f(\w))})))\\=\frac{1}{nc_3^{(s)}}\log(Ee^{c_3^{(s)}\sqrt{n}\min_{\gamma_{str},\nu_{str}^{(1)}
,\nu_{str}^{(2)}\geq 0}(f_1^{str+}(\h,\nu_{str}^{(1)},\nu_{str}^{(2)},\gamma_{str},\beta)+\nu_{str}^{(2)}(2k-n)+\gamma_{str})})\\
\doteq \frac{1}{nc_3^{(s)}}\min_{\gamma_{str},\nu_{str}^{(1)}
,\nu_{str}^{(2)}\geq 0}\log(Ee^{c_3^{(s)}\sqrt{n}(f_1^{str+}(\h,\nu_{str}^{(1)},\nu_{str}^{(2)},\gamma_{str},\beta)+\nu_{str}^{(2)}(2k-n)+\gamma_{str})})\\
=\min_{\gamma_{str},\nu_{str}^{(1)}
,\nu_{str}^{(2)}\geq 0}(\nu_{str}^{(2)}\sqrt{n}(2\beta-1)+ \frac{\gamma_{str}}{\sqrt{n}}+\frac{1}{nc_3^{(s)}}\log(Ee^{c_3^{(s)}\sqrt{n}(f_1^{str+}(\h,\nu_{str}^{(1)},\nu_{str}^{(2)},\gamma_{str},\beta))}))\\
=\min_{\gamma_{str},\nu_{str}^{(1)}
,\nu_{str}^{(2)}\geq 0}(\nu_{str}^{(2)}\sqrt{n}(2\beta-1)+ \frac{\gamma_{str}}{\sqrt{n}}+\frac{1}{nc_3^{(s)}}\log(Ee^{c_3^{(s)}\sqrt{n}(\sum_{i=1}^{n}\t_i^{+})})),\label{eq:gamaiden1strnon}
\end{multline}
where $\t_i^{+}$ is as given in (\ref{eq:streq03non}) and as earlier, $\doteq$ stands for equality when $n\rightarrow \infty$ and would be obtained through the mechanism presented in \cite{SPH} (as in Section \ref{sec:strthr}, for our needs here even just replacing $\doteq$ with an $\leq$ inequality suffices). Now if one sets $\gamma_{str}=\gamma_{str}^{(s)}\sqrt{n}$ and $\nu_{str}^{(2,s)}=\nu_{str}^{(2)}\sqrt{n}$ then (\ref{eq:gamaiden1strnon}) gives
\begin{eqnarray}
I_{str}^{+}(c_3^{(s)},\beta)
& = & \min_{\gamma_{str},\nu_{str}^{(1)}
,\nu_{str}^{(2)}\geq 0}(\nu_{str}^{(2)}\sqrt{n}(2\beta-1)+ \frac{\gamma_{str}}{\sqrt{n}}+\frac{1}{nc_3^{(s)}}\log(Ee^{c_3^{(s)}\sqrt{n}(\sum_{i=1}^{n}\t_i^{+})}))\nonumber \\
& = &
\min_{\gamma_{str}^{(s)},\nu_{str}^{(1)}
,\nu_{str}^{(2,s)}\geq 0}(\nu_{str}^{(2,s)}(2\beta-1)+ \gamma_{str}^{(s)}+\frac{1}{c_3^{(s)}}\log(Ee^{c_3^{(s)}\t_i^{(+,s)}})),\label{eq:gamaiden2strnon}
\end{eqnarray}
where
\begin{equation}
\t_i^{+,s}=\begin{cases}\frac{\h_i^2+(\nu_{str}^{(1,s)})^2}{4\gamma_{str}^{(s)}}-\frac{\h_i\nu_{str}^{(1,s)}}{2\gamma_{str}^{(s)}}-\nu_{str}^{(2,s)}, &  \h_i\leq \nu_{str}^{(1,s)}-\sqrt
{8\gamma_{str}^{(s)}\nu_{str}^{(2,s)}}\\
\nu_{str}^{(2,s)}, & \nu_{str}^{(1,s)}-\sqrt{8\gamma_{str}^{(s)}\nu_{str}^{(2,s)}} \leq \h_i\leq \nu_{str}^{(1,s)}\\
\frac{\h_i^2+(\nu_{str}^{(1,s)})^2}{4\gamma_{str}^{(s)}}-\frac{\h_i\nu_{str}^{(1,s)}}{2\gamma_{str}^{(s)}}+\nu_{str}^{(2,s)}, & \h_i\geq \nu_{str}^{(1,s)}\\
\end{cases}.\label{eq:deftisstrnon}
\end{equation}
The above characterization is then sufficient to compute attainable strong thresholds. However, as in Section \ref{sec:strthr}, since there will be a substantial numerical work involved it is probably a bit more convenient to look for a neater representation. That obviously involves solving a bunch of integrals. We again skip such a tedious job but present the final results. We start with setting
\begin{equation}
I^{+}=Ee^{c_3^{(s)}\t_i^{(+,s)}}.\label{eq:defbigIstrnon}
\end{equation}
Then one has
\begin{equation}
I^{+}=I^{(1+)}+
I^{(2+)}+
I^{(3+)},\label{eq:defbigIstr1non}
\end{equation}
where $I^{(1+)}$, $I^{(2+)}$, and $I^{(3+)}$ are defined below (some of the integrals are similar to the ones we had earlier; however some are different, so we redefine all of them).

\underline{\emph{1. Determining $I^{(1+)}$}}

\noindent Assume that $c_3^{(s)}\geq 0$ and $\gamma_{str}^{(s)}>0$ are constants such that $c_3^{(s)}/4/\gamma_{str}^{(s)}<\frac{1}{2}$ and set
\begin{eqnarray}
p^{+} & = & c_3^{(s)}/4/\gamma_{str}^{(s)}\nonumber \\
q^{+} & = & -c_3^{(s)}\nu_{str}^{(1,s)}/2/\gamma_{str}^{(s)}\nonumber \\
r^{+} & = & c_3^{(s)}((\nu_{str}^{(1,s)})^2/4/\gamma_{str}^{(s)}-\nu_{str}^{(2,s)})\nonumber \\
d^{+} & = & \frac{q^{+}}{\sqrt{2(1-2p^{+})}}\nonumber \\
C_1^{+} & = & e^{(d^{+})^2+r^{+}}/\sqrt{1/2-p^{+}}\nonumber \\
b^{+} &= & (\nu_{str}^{(1,s)}-\sqrt{8\gamma_{str}^{(s)}\nu_{str}^{(2)}})\sqrt{1/2-p^{+}}\nonumber \\
I^{(1+)} & = & =C_1^{+}(2-\mbox{erfc}(b^{+}-d^{+}))\label{eq:defbigI1strnon}
\end{eqnarray}

\underline{\emph{2. Determining $I^{(2+)}$}}

\noindent Set
\begin{equation}
I^{(2+)}=e^{c_3^{(s)}\nu_{str}^{(2,s)}}/2(\mbox{erfc}(b^+/\sqrt{2}/\sqrt{1/2-p^{+}})-\mbox{erfc}(\nu_{str}^{(1,s)}/\sqrt{2}).\label{eq:fbigI2strnon}
\end{equation}

\underline{\emph{3. Determining $I^{(3+)}$}}

\noindent Set
\begin{eqnarray}
r_1^{+} & = & c_3^{(s)}((\nu_{str}^{(1,s)})^2/4/\gamma_{str}^{(s)}+\nu_{str}^{(2,s)})\nonumber \\
C_3^{+} & = & e^{(d^{+})^2+r_1^{+}}/\sqrt{1/2-p^{+}}\nonumber \\
a^{+} &= & \nu_{str}^{(1,s)}\sqrt{1/2-p^{+}}\nonumber \\
I^{(3+)} & = & =C_3^{+}(\mbox{erfc}(a^{+}-d^{+})).\label{eq:defbigI3strnon}
\end{eqnarray}

We summarize the above results related to the strong threshold ($\beta_{str}^{+}$) in the following theorem.

\begin{theorem}(Strong threshold ($\x$ a priori known to be non-negative) -- lifted lower bound)
Let $A$ be an $m\times n$ measurement matrix in (\ref{eq:systemnon})
with i.i.d. standard normal components. Let
the unknown $\x$ in (\ref{eq:systemnon}) be $k$-sparse and let it be a priori known that its nonzero components are positive.
Let $k,m,n$ be large
and let $\alpha=\frac{m}{n}$ and $\betastr^{+}=\frac{k}{n}$ be constants
independent of $m$ and $n$. Let $\mbox{erf}$ be the standard error function associated with zero-mean unit variance Gaussian random variable and let $\mbox{erfc}=1-\mbox{erf}$. Assume that $c_3^{(s)}\geq 0$ and $\gamma_{str}^{(s)}>0$ are constants such that $c_3^{(s)}/4/\gamma_{str}^{(s)}<\frac{1}{2}$.
Let
\begin{equation}
\widehat{\gamma_{sph}^{(s)}}=\frac{2c_3^{(s)}-\sqrt{4(c_3^{(s)})^2+16\alpha}}{8},\label{eq:gamasphthmstrnon}
\end{equation}
and
\begin{equation}
I_{sph}(c_3^{(s)},\alpha)=
\left ( \widehat{\gamma_{sph}^{(s)}}-\frac{\alpha}{2c_3^{(s)}}\log(1-\frac{c_3^{(s)}}{2\widehat{\gamma_{sph}^{(s)}}}\right ).\label{eq:Isphthmstrnon}
\end{equation}
Further, let $I^{+}$ be defined through (\ref{eq:defbigIstr1non})-(\ref{eq:defbigI3strnon}) (or alternatively through (\ref{eq:deftisstrnon}) and (\ref{eq:defbigIstrnon})) and
\begin{equation}
I_{str}(c_3^{(s)},\beta_{str}^{+})=\min_{\gamma_{str}^{(s)}\geq c_3^{(s)}/2,\nu_{str}^{(1)}
,\nu_{str}^{(2,s)}\geq 0}(\nu_{str}^{(2,s)}(2\beta_{str}^{+}-1)+ \gamma_{str}^{(s)}+\frac{1}{c_3^{(s)}}\log(I^{+})).\label{eq:Istrthmstrnon}
\end{equation}
If $\alpha$ and $\betastr^{+}$ are such that
\begin{equation}
\min_{c_3^{(s)}\geq 0}\left (-\frac{c_3^{(s)}}{2}+I_{str}(c_3^{(s)},\beta_{str}^{+})+I_{sph}(c_3^{(s)},\alpha)\right )<0,\label{eq:strcondthmstrnon}
\end{equation}
then the solution of (\ref{eq:l1non}) is with overwhelming
probability the a priori known to be nonnegative $k$-sparse $\x$ in (\ref{eq:systemnon}).\label{thm:thmstrthrliftnon}
\end{theorem}
\begin{proof}
Follows from the above discussion.
\end{proof}

The results for the strong threshold obtained from the above theorem as well as the corresponding ones from \cite{DonohoSigned,DT} and Theorem \ref{thm:thmstrthrnon}
are presented in Figure \ref{fig:strnon} (results obtained based on Theorem \ref{thm:thmstrthrnon} are referred to as direct strong threshold bounds).
\begin{figure}[htb]
\centering
\centerline{\epsfig{figure=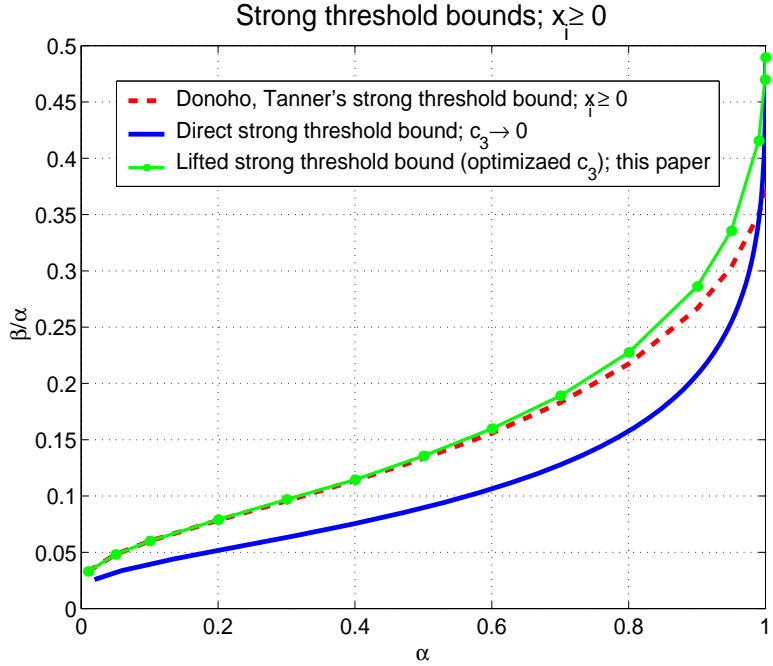,width=10.5cm,height=9cm}}
\caption{\emph{Strong} threshold ($\x_i$ a priori known to be nonnegative), $\ell_1$-optimization}
\label{fig:strnon}
\end{figure}
The results substantially improve on those presented in \cite{StojnicCSetam09}. In fact they improve even on those from \cite{DonohoSigned,DT}. As earlier, since the results are very close to the ones given in \cite{DonohoSigned,DT} we present in Tables \ref{tab:strnontab1} and \ref{tab:strnontab2} the concrete values of the attainable strong thresholds obtained through the mechanisms presented in \cite{DonohoSigned,DT} as well as those obtained through the mechanism presented in the above theorem. We also mention that the results given in Theorem \ref{thm:thmstrthrnon} can in fact be deduced from the above theorem. Namely, in the limit $c_3^{(s)}\rightarrow 0$, one from (\ref{eq:chneg9strnon}) and (\ref{eq:defIsstrnon}) has $\max_{\w\in\Sstrnon}\h^T\w<\sqrt{\alpha n}$ as the limiting condition which is exactly the same condition considered when Theorem \ref{thm:thmstrthrnon} was created. This is not that hard to do but is a fairly tedious procedure and we skip the details.

\begin{corollary}(Strong threshold ($\x$ a priori known to be non-negative) -- lower bound)
Let $A$ be an $m\times n$ measurement matrix in (\ref{eq:system})
with i.i.d. standard normal components. Let
the unknown $\x$ in (\ref{eq:systemnon}) be $k$-sparse and let it be a priori known that its nonzero components are positive. Let $k,m,n$ be large
and let $\alpha^+=\frac{m}{n}$ and $\beta_{str}^{+}=\frac{k}{n}$ be constants
independent of $m$ and $n$. Let $\mbox{erf}$ and $\erfinv$ be the standard error function and its inverse associated with zero-mean unit variance Gaussian random variable. Also let $\mbox{erfc}=1-\mbox{erf}$. Set
\begin{eqnarray}
S_{1}^{(+,0)} & = & \frac{1}{2}\mbox{erfc}(\frac{\nu_{str}^{(1)}}{\sqrt{2}})+\frac{\nu_{str}^{(1)}}{\sqrt{2\pi}}e^{-\frac{\nu_{str}^{(1)}}{2}}\nonumber \\
S_{2}^{(+,0)} & = & 1/2/\sqrt{2\pi}\left (2\sqrt{2\pi}\beta_{str}^{+}+\frac{2\sqrt{2}(\mbox{erfinv}(2(1-\beta_{str}^{+})-1))}
{e^{(\mbox{erfinv}(2(1-\beta_{str}^{+})-1))^2}}\right )\nonumber \\
S_{3}^{(+,0)} & = & (\frac{1}{2}\mbox{erfc}(\frac{\nu_{str}^{(1)}}{\sqrt{2}})+\beta_{str}^{+})(\nu_{str}^{(1)})^2+\nu_{str}^{(1)}\sqrt{\frac{2}{\pi}}(e^{-(\mbox{erfinv}(2(1-\beta_{str}^{+})-1))^2}-e^{-\frac{\nu_{str}^{(1)}}{2}}).\label{eq:cordefS1S2S3strthetanon}
\end{eqnarray}
If $\alpha^+$ and $\beta_{str}^{+}$ further satisfy
\begin{equation}
\alpha^{+}>\min_{\nu_{str}^{(1)}\geq 0}(S_{1}^{(+,0)}+S_{2}^{(+,0)}+S_{3}^{(+,0)})\label{eq:corstralphanon}
\end{equation}
then with overwhelming probability the solution of (\ref{eq:l1non}) is the $k$-sparse $\x$ from (\ref{eq:systemnon}).\label{cor:corstrthrnon}
\end{corollary}

\begin{proof}
Theorem \ref{thm:thmstrthrliftnon} holds for any $c_3^{(s)}\geq 0$. The above corollary instead of looking for the best possible $c_3^{(s)}$ in Theorem \ref{thm:thmstrthrliftnon} assumes a simple $c_3^{(s)}\rightarrow 0$ scenario.

Alternatively, one can look at $\frac{E\max_{\w\in\Sstrnon}\h^T\w}{\sqrt{n}}$ and following the methodology presented in (\ref{eq:seceq1}), (\ref{eq:streq01}), and  (\ref{eq:streq02}) (and originally in \cite{StojnicCSetam09}) obtain for $\nu_{str}^{(1)}\geq 0$
\begin{equation}
\frac{E\max_{\w\in\Sstrnon}\h^T\w}{\sqrt{n}}\leq \sqrt{E_{\h_i\leq c_{\nu}^+}(\h_i-\nu_{str}^{(1)})^2+ E_{\nu_{str}^{(1)}\leq \h_i}(\h_i-\nu_{str}^{(1)})^2},\label{eq:corrstr0non}
\end{equation}
where $c_{\nu}^+$ is obtained from $\beta_{str}^+=\int_{\h_i\leq c_{\nu}^+}\frac{e^{-\frac{\h_i^2}{2}}d\h_i}{\sqrt{2\pi}}$. Clearly, $c_{\nu}^+=-\sqrt{2}\mbox{erfinv}(2(1-\beta_{str}^+)-1)$.
Optimizing (tightening) over $\nu_{str}^{(1)}\geq 0$ (and using all the concentrating machinery of \cite{StojnicCSetam09}) one can write
\begin{equation}
\frac{E\max_{\w\in\Sstr}\h^T\w}{\sqrt{n}}\doteq\min_{\nu_{str}^{(1)}\geq 0}\sqrt{\left (\int_{\h_i\leq c_{\nu}^+}(\h_i-\nu_{str}^{(1)})^2\frac{e^{-\frac{\h_i^2}{2}}d\h_i}{\sqrt{2\pi}}
+\int_{\nu_{str}^{(1)}\leq \h_i}(\h_i-\nu_{str}^{(1)})^2\frac{e^{-\frac{\h_i^2}{2}}d\h_i}{\sqrt{2\pi}}\right )}.\label{eq:corrstr1non}
\end{equation}

On the other hand, when $c_3^{(s)}\rightarrow 0$ a combination of (\ref{eq:Isphcorsec}), (\ref{eq:streq030non}), (\ref{eq:gamaiden2strnon}), and (\ref{eq:strcondthmstrnon}) gives
\begin{equation}
\min_{\gamma_{str}^{(s)},\nu_{str}^{(1)}
,\nu_{str}^{(2,s)}\geq 0}(\nu_{str}^{(2,s)}(2\beta_{str}^+-1)+ \gamma_{str}^{(s)}+E\t_i^{(+,s)})\leq \sqrt{\alpha},\label{eq:gamaiden2strcorprnon}
\end{equation}
where
\begin{equation}
\t_i^{(+,s)}=\max\left (\frac{(\h_i-\nu_{str}^{(1)})^2}{4\gamma_{str}^{(s)}}-\nu_{str}^{(2,s)},
\frac{(\max(\h_i-\nu_{str}^{(1)},0))^2}{4\gamma_{str}^{(s)}}+\nu_{str}^{(2,s)}\right ).\label{eq:streq030corprnon}
\end{equation}
A particular choice $\nu_{str}^{(2,s)}=\frac{(c_{\nu}^+-\nu_{str}^{(1)})^2}{8\gamma_{str}^{(s)}}$ (where as above, $c_{\nu}^+=-\sqrt{2}\mbox{erfinv}(2(1-\beta_{str}^+)-1)$) transforms (\ref{eq:gamaiden2strcorprnon}) to
\begin{equation}
\min_{\gamma_{str}^{(s)},\nu_{str}^{(1)}\geq 0}
\left (\gamma_{str}^{(s)}+\frac{\left (\int_{\h_i\leq c_{\nu}^+}(\h_i-\nu_{str}^{(1)})^2\frac{e^{-\frac{\h_i^2}{2}}d\h_i}{\sqrt{2\pi}}
+\int_{\nu_{str}^{(1)}\leq \h_i}(\h_i-\nu_{str}^{(1)})^2\frac{e^{-\frac{\h_i^2}{2}}d\h_i}{\sqrt{2\pi}}\right )}{4\gamma_{str}^{(s)}}\right )\leq \sqrt{\alpha},\label{eq:gamaiden2strcorpr1non}
\end{equation}
which after optimization over $\gamma_{str}^{(s)}$ becomes
\begin{equation}
\min_{\nu_{str}^{(1)}\geq 0}
\sqrt{\left (\int_{\h_i\leq c_{\nu}^+}(\h_i-\nu_{str}^{(1)})^2\frac{e^{-\frac{\h_i^2}{2}}d\h_i}{\sqrt{2\pi}}
+\int_{\nu_{str}^{(1)}\leq \h_i}(\h_i-\nu_{str}^{(1)})^2\frac{e^{-\frac{\h_i^2}{2}}d\h_i}{\sqrt{2\pi}}\right )}\leq \sqrt{\alpha}.\label{eq:gamaiden2strcorpr2non}
\end{equation}
This is exactly what one gets from (\ref{eq:corrstr1non}). Solving the integrals from (\ref{eq:corrstr1non}) or (\ref{eq:gamaiden2strcorpr2non}) then leads to the condition given in the above corollary. At the same time this also shows that in (\ref{eq:streq04non}) one, for all statistical purposes, indeed has an equality (as earlier, much more is true but the proof of that would just make our presentation more cumbersome and those statements are not necessary in a statistical context that we consider here).
\end{proof}

\noindent \textbf{Remark:} Solving (\ref{eq:corstralphanon}) over $\nu_{str}^{(1)}$ and juggling a bit one can arrive at the results presented in Theorem \ref{thm:thmstrthrnon}. Of course, the results of Theorem \ref{thm:thmstrthrnon} were presented in a way that parallels other results from \cite{StojnicCSetam09} where our main concern was a thorough studying of the underlying optimization problems.

As earlier, the results obtained from the previous corollary (and obviously the results from Theorem \ref{thm:thmstrthrliftnon}) are those presented in Figure \ref{fig:strnon} and Tables \ref{tab:strnontab1} and \ref{tab:strnontab2} that we refer to as $c_3\rightarrow 0$ scenario or direct strong threshold bounds (also, as earlier, we remove superscript $(s)$ from $c_3^{(s)}$ to make the figure and tables easier to view). As can be see from Figure \ref{fig:strnon} and Tables \ref{tab:strnontab1} and \ref{tab:strnontab2}, a substantial improvement over the strong threshold results from \cite{StojnicCSetam09} can be achieved through the above introduced mechanism.

Any improvement in strong thresholds is typically hard to achieve. In that regard the improvement we created over the results one could obtain through the mechanism of \cite{StojnicCSetam09} (essentially those presented in Theorem \ref{thm:thmstrthrnon}) is welcome. If one compares the results from Theorem \ref{thm:thmstrthrliftnon} to those obtained in \cite{DonohoUnsigned,DT} (which are the best known lower bounds) the improvement is not as big. As was the case earlier, it is quite possible that the results of \cite{DonohoUnsigned,DT} are very close to the optimal ones. Also, it is possible that neither our results nor those of \cite{DonohoUnsigned,DT} are even remotely close to the optimal ones in which case the improvements we presented above are conceptually substantial when compared to what can be done through the mechanism of \cite{StojnicCSetam09} but are a bit powerless when it comes to overall optimality (this seems a bit unlikely though). On the other hand our work included a substantial amount of numerical computations. There is about a billion places where something could be off by a bit. First, just missing one constant here or there would be enough to completely mess up everything. Second, all of computations are naturally done only with a finite precision. Third, all numerical solutions are done on a local optimum level. Still, we are inclined to believe that the methods we presented indeed can not achieve substantially more than what we showed in Tables \ref{tab:strnontab1} and \ref{tab:strnontab2} and Figure \ref{fig:strnon} (maybe a little bit more here and there due to a potential lack of global optimality and numerical precision but overall really not that much). As for how far away from the optimal are the results we presented, we again do believe that the curve is fairly close to the optimal one.

\begin{table}
\caption{Strong threshold bounds; non-negative $\x$ -- low $\alpha\leq 0.5$ regime}\vspace{.1in}
\hspace{-0in}\centering
\begin{tabular}{||c|c|c|c|c|c|c|c||}\hline\hline
 $\alpha$  & $0.01$ & $0.05$ & $0.1$ & $0.2$ & $0.3$ & $0.4$ & $0.5$ \\ \hline\hline
 $\beta_{str}^{+}$ (Donoho, Tanner \cite{DonohoUnsigned,DT}) & $0.00033$ & $0.0024$ & $0.0060$ & $0.0157$ & $0.0287$ & $0.0455$ & $0.0667$ \\ \hline
 $\beta_{str}^{+}$ (Theorem \ref{thm:thmstrthrliftnon})
 & $0.00033 $ & $0.0024 $ & $0.0060 $ & $0.0158 $ & $0.0291 $ & $0.0461 $ & $0.0680$ \\ \hline\hline
\end{tabular}
\label{tab:strnontab1}
\end{table}

\begin{table}
\caption{Strong threshold bounds; non-negative $\x$ -- high $\alpha> 0.5$ regime}\vspace{.1in}
\hspace{-0in}\centering
\begin{tabular}{||c|c|c|c|c|c|c|c|c||}\hline\hline
 $\alpha$  & $0.6$ & $0.7$ & $0.8$ & $0.9$ & $0.95$ & $0.99$ & $0.999$ & $0.9999$ \\ \hline\hline
 $\beta_{str}^{+}$ (Donoho, Tanner \cite{DonohoUnsigned,DT})  & $0.0935$ & $0.1280$ & $0.1739$ & $0.2399$ & $0.2881$ & $0.3463$ & $0.3675$ & $0.3750$ \\ \hline
 $\beta_{str}^{+}$ (Theorem \ref{thm:thmstrthrliftnon})
& $0.0959$ & $0.1323
$ & $0.1820 $ & $0.2577 $ & $0.3188$ & $0.4113$ & $0.4694$ & $0.4895$\\ \hline\hline
\end{tabular}
\label{tab:strnontab2}
\end{table}

\section{Conclusion}
\label{sec:conc}

In this paper we looked at classical under-determined linear systems with sparse solutions. We analyzed a particular optimization technique called $\ell_1$ optimization. While the technique is known to work well often its ultimate performance limits in certain scenarios are not known. We attacked a couple of such scenarios, namely those that relate to what is called the strong and sectional thresholds. We developed a couple of mechanisms that can be utilized together with our recent results from \cite{StojnicMoreSophHopBnds10} related to a couple of general classical combinatorial optimization problems. The mechanisms that we developed provided a substantial improvement over their counterparts from \cite{StojnicCSetam09}. Moreover, in a wide range of problem parameters (dimensions) we feel confident that the results we obtained are actually fairly close to the exact ones.

To be a bit more specific, we provided purely theoretical lower bounds on the values of the strong and sectional thresholds of $\ell_1$ minimization technique. These bounds were first created for problems where unknown vectors $\x$ are general sparse vectors from $R^n$. We then in the second part of the paper adapted them so that they can fit the case when the unknown vectors $\x$ are a priori known to be nonnegative (obviously, due to the signed structure of the nonnegative vectors the sectional counterparts do not apply for them). All results we presented offer a substantial conceptual improvement over their counterparts from \cite{StojnicCSetam09}. In a majority of cases and in a fairly wide range of parameters they actually provide even practically important improvements. We also, demonstrated how the results of the core mechanism that we utilized in \cite{StojnicCSetam09} can be deduced from the ones we presented in this paper.

As was the case in \cite{StojnicCSetam09,StojnicMoreSophHopBnds10}, the purely theoretical results we presented are for the so-called Gaussian models, i.e. for systems with i.i.d. Gaussian coefficients. Such an assumption significantly simplified our exposition. However, all results that we presented can easily be extended to the case of many other models of randomness. There are many ways how this can be done. Instead of recalling on them here we refer to a brief discussion about it that we presented in \cite{StojnicMoreSophHopBnds10}.

As for usefulness of the presented results, there is hardly any limit. First, one can look at a host of related problems from the compressed sensing literature. Pretty much any problem we attacked in a weak sense (and there is pretty much none solvable in polynomial time that we couldn't solve to an ultimate precision) can now be attacked in strong/sectional sense as well. These include for example, all noisy variations, approximately sparse unknown vectors, vectors with a priori known structure (block-sparse, binary/box constrained etc.), vectors with a priori partially known support, isometry constants, all types of low rank matrix recoveries, various other algorithms like $\ell_q$-optimization, SOCP's, LASSO's, and many, many others. Each of these problems has its own specificities and adapting the methodology presented here usually takes a bit of work but in our view is now a routine. While we will present some of these applications we should emphasize that their contribution will be purely on an application level. As we already mentioned many times, the mathematical core of the arguments is here, in \cite{StojnicMoreSophHopBnds10}, and to a degree in \cite{StojnicRegRndDlt10}.

Of course, way beyond the concrete threshold results related to compressed sensing problems and various other duality type of optimization problems that we presented here and in quite a few companion papers, is the universal value of the presented mechanisms that especially becomes visible when one starts encountering problems that are believed not to be solvable in polynomial time. Namely, what we presented here (and in a nutshell in \cite{StojnicMoreSophHopBnds10}) is an incredibly powerful concept to attack the statistical behavior of many hard combinatorial problems. The mechanisms that we developed in \cite{StojnicCSetam09,StojnicUpper10,StojnicRegRndDlt10} are powerful enough to handle to an ultimate precision statistical behavior of a huge subclass of optimization problems where the duality holds. However, where it comes to those where the duality fades away (and many combinatorial ones are such examples) the methods from \cite{StojnicCSetam09,StojnicUpper10,StojnicRegRndDlt10} can only provide bounds on the typical behavior. In that sense not much more is done here and in \cite{StojnicMoreSophHopBnds10}. Quite contrary, what was provided here and in \cite{StojnicMoreSophHopBnds10} are also just bounds on the typical behavior of the hard random combinatorial optimization problems. Of course, such a view would be a huge understatement. In our own experience, even establishing that mechanisms from \cite{StojnicCSetam09,StojnicUpper10,StojnicRegRndDlt10} can be used to provide rigorous lower bounds for combinatorial problems was not so trivial before we discovered it. Improving on them (and for that matter on any other type of bounds when it comes to combinatorial problems) is typically super hard. The mechanisms that we presented here (and partially in \cite{StojnicMoreSophHopBnds10}) essentially provide an avenue for attacking a huge number of other combinatorial problems. Many of them we attacked and achieved a substantial improvement over known results.
Examples include, numerous variants of knapsack, max-cut, subset selection, bin-packing, number-partitioning problems, binary/spherical perceptrons, capacities of associative memories and many, many others.
We will present these applications elsewhere. However, we do want to emphasize once again, right here that the core of all of these mechanisms that enabled us to achieve such substantial improvements in studying many combinatorial problems is actually precisely what we presented above, in \cite{StojnicMoreSophHopBnds10}, and to a degree in \cite{StojnicRegRndDlt10}.

As for attacking problems considered here, namely, the strong and sectional thresholds of $\ell_1$-minimization, what we presented above is of course only one way how one can do it. There are several other ones that we developed that are quite likely more appropriate for these particular problems. However, we decided to present this one since we believe it has a general mathematical value that massively supersedes its linear systems/$\ell_1$/compressed sensing importance. Its general value is exactly what we described above, i.e. its ability to lift typical convexity type of bounds the mechanisms of \cite{StojnicCSetam09} would provide when applied to combinatorial problems. Also, the problems studied here were among a few very first ones that we attacked with the mechanism so we felt that it would be appropriate to present the methods the way that tightly follows their chronological development.

Of course, all of what we presented here and especially of what we presented in the original introductory paper \cite{StojnicMoreSophHopBnds10} can be shown to be  tightly connected with concepts from statistical physics as well. Namely, our results from e.g. \cite{StojnicCSetam09,StojnicUpper10,StojnicRegRndDlt10} provide the exact analysis for the optimization problems where the duality holds. Since they are the same as the results produced by what is in statistical physics called the $0$-level of replica symmetry breaking, our results essentially rigorously establish that the replica symmetry results are typically correct in optimization problems where the duality holds (a fact long believed to be true among physicists). Moreover, even when the duality does not hold they rigorously show that the replica symmetry results are typically (essentially if derived in a certain way) rigorous lower bounds. On the other hand, the results we presented here and especially those we presented initially in \cite{StojnicMoreSophHopBnds10} go way further. Namely, one can show that they match the corresponding ones obtained through a variant of 1-level of replica symmetry breaking which then establishes the latter ones also as rigorous lower bounds (these are for many combinatorial problems typically already very close to the optimal values). As mentioned above, for this to be transparent one would have to derive the replica breaking in a certain way. That is not that hard but is not super obvious either. We find it useful to have such a connection neatly established and discussed in more detail; however, since it is a bit technical and requires quite a bit of detailing we will present it separately in a forthcoming paper. Here, we do mention that Gordon's introductory examples from \cite{Gordon85} are essentially the first sets of results that establish rigorous lower bounds on the type of problems that fit the negative Hopfield models (our results in \cite{StojnicHopBnds10} just recognized that some of those or their small alterations are applicable within the frame of statistical physics and in fact do match what is typically called replica symmetry level ($0$-breaking level) solutions thereby establishing such solutions as rigorous bounds). It is the main breakthrough of \cite{StojnicMoreSophHopBnds10} (and now this paper) though, that actually enables us to go beyond simple duality/convexity $0$-level (replica symmetry type of) bounds.

While we will present a detailed connection with the replica theory elsewhere, we would just briefly like to mention here a few things that relate to such a connection. As said above, establishing a connection with the replica analysis is not that hard. While proving replica procedures rigorously on the current axiomatic system of mathematics is hopeless (if nothing else, at the very least due to needed continuity of algebraic dimensions which does not exist in typical algebras we use), proving many results it provides is likely possible. Since our results are purely mathematical and completely rigorous, their connection with the results of the replica theory then makes the latter ones rigorous as well. However, proving rigorously any of this type of results is very hard and obviously all results that we presented above in that direction are a consequence of a massive effort that was put fort to create mathematically sound concepts. Along the same lines, one then wonders, if we can establish the connection up to a variant of 1-level of replica breaking can't we go further and present our mechanism so that it would correspond to as many levels of replica breaking as one wants. Well the answer is yes, one can easily (to be more precise, the mathematical logic is relatively simple but the writing is not simple at all!) write down the formalism and the entire cascade of exponentials which would resemble what one would get in the replica theory (of course for that one would have to derive the replica concept in certain way). However, even after doing all of that one has to be very careful in establishing needed inequalities and relations between coefficients and occasionally may need to adjust the form of the objective function so that all the inequalities hold. This is a very tedious task and requires quite a lot of technicalities (especially if one also wants to show all the concentrations as well) even if it turns out that it leads to rigorous optimal solutions.

Of course, the above would be just to establish the lower bounds; proving that they are optimal (which we can't even guess if it is true or not) would be even more horrendous. One can just take a look at \cite{Tal06} for the upper-bounding effort in the much easier SK model to get a feeling what kind of treat we could be in for (of course, we do mention that such a method is highly likely not to be applicable for the models we studied here and in the case of the negative form from \cite{StojnicMoreSophHopBnds10}; this, of course, is one of the main reasons why the Hopfield models and their more general versions studied here are believed to be much harder mathematical challenges than the original Parisi formula and the SK model it treats).

Going back to raising the levels of nested expectations, we believe that the final results of such a cascade approach would be almost no different than those we presented in figures and tables in this paper (in fact, limited numerical experiments that we did indicate that for some values of $\alpha$ further improvements on $\beta$ are literally on the third, or quite often even fourth, significant decimal digit; of course this has to be taken somewhat lightly; first we didn't examine the entire range of $\alpha$'s and second, one has to keep in mind that already on the second nesting level there may be $6-7$ variables to optimize simultaneously combined with a numerical integration as well which easily can render numerical errors either way). Obviously, it may happen that the optimal cascade is actually of infinite length, but two-three levels should already be enough to achieve a decent convergence. Of course, if one is so lucky it can happen that optimal length is actually one (this is highly unlikely to be true at least for the original positive and negative Hopfield models from \cite{StojnicMoreSophHopBnds10}). Either way, it is our belief that the major improvements are on the first level which is what we presented (we also emphasize that although our presentation is on spots quite involved, it is about billion times simpler than what happens on higher levels). Since all of that takes an incredible technical/numerical effort, as mentioned above, we will discuss it elsewhere.

\begin{singlespace}
\bibliographystyle{plain}
\bibliography{LiftThrBndsRefs}
\end{singlespace}

\end{document}